\documentclass[journal]{IEEEtran}
\pdfoutput=1

\usepackage{cite}
\usepackage{amsmath, amsthm, amssymb}
\usepackage{epsfig}
\usepackage{epsf}
\usepackage{psfrag}
\usepackage{amsfonts}
\usepackage{graphicx}
\usepackage{hyperref}	
\usepackage{setspace}

\usepackage{lipsum}
\usepackage{float}
\usepackage{psfrag}
\usepackage{color}
\usepackage{subfigure}
\usepackage{mathrsfs}
\usepackage{indentfirst}
\usepackage{fancyvrb}
\usepackage{algorithm}
\usepackage{algpseudocode,algorithmicx}
\usepackage{transparent}

\newtheorem{theorem}{Theorem}
\newtheorem{proposition}{Proposition}
\theoremstyle{definition}
\newtheorem{definition}{Definition}[section]
\newtheorem{lemma}{Lemma}

\newtheorem{example}{Example}

\DeclareMathOperator*{\argmin}{argmin}
\DeclareMathOperator*{\argmax}{argmax}

\begin{document}

\title{Epidemic Source Detection in Contact Tracing Networks: Epidemic Centrality in Graphs and Message-Passing Algorithms}

\author{
  \IEEEauthorblockN{Pei-Duo Yu$^\ddag$, Chee~Wei~Tan$^\dag$ and Hung-Lin Fu$^\ast$\\}
  \IEEEauthorblockA{Chung Yuan Christian University$^\ddag$, City University of Hong Kong$^\dag$, National Yang Ming Chiao Tung University$^\ast$ \\
 peiduoyu@cycu.edu.tw, cheewtan@gmail.com, hlfu@math.nctu.edu.tw }
\thanks{This work is supported in part by the Ministry of Science and Technology of Taiwan under Grant 110-2115-M-033-001-MY2, Hong Kong ITF Project ITS/188/20 and an Institute for Pure and Applied Mathematics fellowship.}
}

\maketitle

\begin{abstract}
We study the epidemic source detection problem in contact tracing networks modeled as a graph-constrained maximum likelihood estimation problem using the susceptible-infected model in epidemiology. Based on a snapshot observation of the infection subgraph, we first study finite degree regular graphs and regular graphs with cycles separately, thereby establishing a mathematical equivalence in maximal likelihood ratio between the case of finite acyclic graphs and that of cyclic graphs. In particular, we show that the optimal solution of the maximum likelihood estimator can be refined to distances on graphs based on a novel statistical distance centrality that captures the optimality of the nonconvex problem. An efficient contact tracing algorithm is then proposed to solve the general case of finite degree-regular graphs with multiple cycles. Our performance evaluation on a variety of graphs shows that our algorithms outperform the existing state-of-the-art heuristics using contact tracing data from the SARS-CoV 2003 and COVID-19 pandemics by correctly identifying the superspreaders on some of the largest superspreading infection clusters in Singapore and Taiwan.
\end{abstract}


\section{Introduction}\label{introduction}

The COVID-19 coronavirus pandemic has revealed severe deficiencies in public health protection \cite{fei2021overview}. As the COVID-19 disease is highly contagious and wide-ranging with long incubation periods (transmission rate of 3-5 persons within 6 feet), it becomes necessary to track down all the infected persons and their recent contacts once an outbreak has occurred. It is often necessary to account for the initial source of the outbreak (e.g., identification of superspreaders \cite{kemper1980identification,stein2010}) so that public health is resilient against further outbreaks or to understand the underlying cause of secondary transmissions. Public health authorities worldwide employ contact tracing to address these problems \cite{ct1,ct2,plosone,naturephysics2021,cencetti2021digital,bradshaw2021bidirectional,bengio2020predicting,fei2021overview}. 

In general, manual contact tracing is a complex and tedious process: Once a person has been diagnosed as infected, public health authorities fan out to trace the recent contacts of this person for the purpose of monitoring or quarantine. This process repeats if one of those contacts exhibits symptoms until all the contacts who have been exposed are out of circulation. The COVID-19 coronavirus epidemic has overwhelmed most contact tracing capabilities due to the speed and scale of infection \cite{naturephysics2021}. For contact tracing to be efficient, especially to identify superspreaders in recurrent large-scale outbreaks during the COVID-19 pandemic, efficient and scalable contact tracing algorithms need to be developed \cite{ct1,ct2,plosone,naturephysics2021,cencetti2021digital,bradshaw2021bidirectional,bengio2020predicting,fei2021overview}.

Recently, there are public health protection schemes that employ mobile software technologies to use wireless signals (e.g., Bluetooth) to collect data on social contact connectivity for {\it digital contact tracing} in epidemiology \cite{ct1,ct3,bengio2020predicting}. These data lead to contact tracing networks that are essentially large graphs generated by stochastic processes. From an epidemiology perspective, all the infected persons in a contact tracing graph are potential candidates for tracking purposes as well as identification of Patient Zero or superspreaders \cite{kemper1980identification,stein2010}. A fundamental question in digital contact tracing is {\it how to unravel stochastic spreading processes to find the initial outbreak source quickly, accurately and reliably with high confidence by exploiting the topological and statistical properties of contact tracing networks.} 

The spreading of epidemics and rumors share much in common as stochastic processes in mathematical epidemiology \cite{Bailey_sir,epidemicsrumors}. Formulating the source detection as solving a maximum likelihood estimation problem was first studied in \cite{who} using a network centrality called the \textit{rumor centrality} to optimally solve a special case of degree-regular tree graph with countably infinite number of vertices and assuming a $\mathsf{SI}$ (susceptible-infectious) spreading model. There were a number of problem extensions subsequently, e.g., random increasing trees in \cite{fuch}, probabilistic sampling in \cite{kar1}, star graph topology in \cite{sri_star}, multiple sources or observations in \cite{tay, BAPrak,rumor_jstsp_2015}, Markov chain Monte Carlo based algorithms in \cite{kar2} and probabilistic characterization of infection network boundaries in \cite{LZhengDSP}. Source estimation for other spreading models include the Susceptible-Infected-Recovered model in \cite{ZhuLei} and a random branching process for irregular trees in \cite{peter1,fuch}. Related work on proactive network protection against epidemics include \cite{rumor_jstsp_2018,pinto, paluch}, e.g., a maximum likelihood algorithm in \cite{paluch} that complements the work in \cite{pinto}. How to optimally place observations concerning the efficiency and the cost and design message-passing algorithms were studied in \cite{rumor_jstsp_2018}.  Probabilistic inference methods that use machine learning techniques like graph neural networks and deep learning for epidemiology have been studied in \cite{murphy2021deep,siya2021CISS, bengio2020predicting,fei2021overview}. 

To the best of our knowledge, there is no prior work that solves the problem in \cite{who} when the network has cycles or finite graph boundaries even under the $\mathsf{SI}$ spreading model. This more general problem is computationally challenging due to the presence of \textit{irregular vertices} (vertices without susceptible neighbors) and cycles. Irregular vertices are practical for modeling actual graph connectivity or users who are quarantined after being infected. In addition, the presence of cycles cannot be ignored. In essence, unlike the tree graphs, the cycles allow the stochastic spreading process to spread through multiple alternate paths, thus increasing the likelihood of those vertices on the cycle. As such, the presence of cycles and irregular vertices introduce nontrivial irregular effects that significantly shape the way that virus spread as well as the performance of detecting the original source on the infected graph. To be exact, existing algorithms in the literature, e.g., \cite{shah2,who,tan1,tan2}, are no longer optimal even with the presence of a single cycle in degree-regular pseudo-tree graphs.

In this paper, we address the outstanding issues of epidemic source detection for the general cases when there are irregular vertices and cycles. In particular, we demonstrate that the likelihood ratios between vertices of the infection graph provide insights to locating the most likely source, which leads to a new network centrality connecting graph-theoretic structures with estimation performance and enables low-complexity algorithms using the same spreading model as those studied in the literature \cite{shah2,who,loh2, tan1, kar1}. 

\begin{table*}
\caption{We summarize the main results of this paper with a comparison with prior art in \cite{who} in this table. Assume the underlying network $G$ is a degree-regular tree. We denote the irregular vertex as $v_{ir}$ and the epidemic center as $v_c$. Note that, the case that $G_n$ containing an irregular vertex can be treated mathematically as a special case of $G_n$ containing a single cycle, since an irregular vertex can be treated as a size one cycle which is on the boundary of the graph. Therefore, the logical flow of this paper is such that we first consider the irregular effect caused by an irregular vertex and then analysis to graphs with cycles.}
\label{tab:result_conclusion}
\centering
\begin{tabular}{c|c|c|c}
Underlying model assumption & Infinite-size $G$  & Finite-size $G$ & $G$ as Pseudo-Tree (cf. Definition \ref{def:pseudo_tree}) \\
\hline
An irregular vertex in $G_n$ & No & Yes & No \\
\hline
A cycle in $G_n$ & No & No & Yes \\
\hline
 MLE & $v_c$ & On the path from $v_c$ to $v_{ir}$ &  On the path from $v_c$ to $C_h$ \\
\hline
Key theorems in this paper & -- & Theorems \ref{thm:longterm}, \ref{thm:keyfinite}& Theorems \ref{thm:rough_locate_cycle}, \ref{thm:MLE_cycle}
\end{tabular}
\end{table*}

\subsection{Our Contributions}
The main contributions are summarized as follows:
\begin{itemize}
    \item  We consider a general network (i.e., finite large network with cycles) for epidemic source detection, and analytically characterize how graph distances between vertices in an infection graph and a single cycle or irregular vertex affect the likelihood of each vertex being the source. 
    
    \item For a finite size degree-regular tree, we characterize the globally optimal maximum likelihood estimator based on graph distances and propose a scalable message-passing algorithm to compute this solution in linear time.
    
    \item For a degree-regular graph with cycles, we prove that its epidemic center can be equal to the epidemic center of any of its spanning trees. We propose a polynomial-time algorithm to find the epidemic center of a graph with a single cycle by characterizing the optimal maximum likelihood estimator based on the graph distances.

    \item We combine the results in finite-size degree regular graphs and unicyclic graphs and propose a polynomial-time distance-based algorithm to find the source estimator. For synthetic data, we conduct simulations on two finite-size regular graphs with cycles which are grid graphs and circulant graphs. We evaluate the performance by comparing the results with the $\mathsf{BFS}$ heuristic rumor centrality \cite{who}, showing that the error (number of hops) of our source estimator is at least $50\%$ smaller than $\mathsf{BFS}$ rumor center. Using real-world contact tracing data (due to SARS-CoV 2003 and COVID-19 pandemic), we apply our algorithm on some of the largest superspreading infection clusters in Singapore and Taiwan, correctly identifying the superspreader in this cluster.
\end{itemize}

\section{Preliminaries on Virus Spreading Model}
\label{infinite}
We model a contact network by an undirected graph $G=(V,E)$, where the set of vertices $V$ represents the vertices in the underlying network, and the set of edges $E$ represents the links between the  vertices. We shall assume that $V$ is {\it countably finite} (this is the crucial departing point from the previous assumption of infinite graph in the literature \cite{shah2,who,tay,tan1,tan2}). In this paper, we use the Susceptible-Infectious (SI) model in \cite{Bailey_sir,epidemicsrumors,BB} to model virus spreading. Vertices that infected by the virus are called {\it infected vertices} and otherwise they are  {\it susceptible vertices}.  The spreading is initiated by a single vertex $v^{\star}\in V$ that we call the source. Once a vertex is infected, it stays infected and can in turn infect its susceptible neighbors. The virus can be spread from vertex $i$ to vertex $j$ if and only if there is an edge between them (i.e., $(i,j)\in E$). Let $\tau_{ij}$ be the spreading time from $i$ to $j$, which are random variables that are independently and exponentially distributed with parameter $\lambda$ (without loss of generality, let $\lambda=1$). Let $\mathbf{S}$ denote the set of all susceptible vertices that have at least one infected neighbor, i.e., those vertices in $\mathbf{S}$ might be infected in the near future. In the real world there are some people that are more likely to be infected by the virus, and some are more likely to spread the virus to others. We can assume that each person has two parameters say $R_{\mathsf{i}}$ and $R_{\mathsf{s}}$ which are corresponding to the rate of being infected and the rate of spreading the virus to others respectively. Due to the memoryless property of the exponential distribution. We have the fact that each newly infected vertex $v$ is randomly chosen from $\mathbf{S}$ with the probability that 
\begin{center}
    $ P(v \text{ is infected}) \propto R_{\mathsf{i}}^v\cdot \sum\limits_{u} R_{\mathsf{s}}^u$,
\end{center}
where each $u$ is an infected neighbor of $v$, $R_{\mathsf{i}}^v$ denote the infected rate of $v$ and $R_{\mathsf{s}}^u$ denote the spreading rate of $u$. Hence, the probability of a vertex $v_a$ being infected in the next time period is defined as
\begin{equation}
\label{eq:spread_p}
    P(v_a \text{ is infected})=\frac{R^{v_a}_{\mathsf{i}}\cdot \sum\limits_{u_a} R_{\mathsf{s}}^{u_a}}{\sum\limits_{v\in \mathbf{S}}[R^{v}_{\mathsf{i}}\cdot \sum\limits_{u} R_{\mathsf{s}}^u ]},
\end{equation}
where $u_a$ and $u$ represents each infected neighbor of $v_a$ and $v$ respectively.

Now, we have a random spreading model over an underlying {\it finite graph} $G$. We can view $G$ as a {\it contact network} \cite{Salathe22020,eletreby2020effects,harvardreview,chen2020time} where nodes represent individuals and edges represent social contacts such as two individuals have been in the same place or in close contact (within about $6$ feet). Contact tracing is a process to identify the people who have been in contact with an infected individual \cite{bengio2020predicting}. Hence, an epidemic contact tracing network can be seen as a connected subgraph of a human contact network containing infected people.  Let $G_n$ be a subgraph of order $n$ of $G$, that models a snapshot observation of the spreading when there are $n$ infected vertices, i.e., $G_n$ is a contact tracing network in $G$ and $\vert G_n\vert =n$. In the following, we shall call $G_n$ an infected subgraph of $G$ for brevity. We denote the {\it actual source} in $G_n$ as $v^{\star}$. The epidemic source detection problem is thus to find $v^{\star}$ given this observation of $G_n$. For example, in Fig. \ref{fig:irregular_ex} the infected subgraph are those vertices labeled from one $1$ to $6$, i.e., $V(G_n)=\{v_1, v_2, \ldots, v_6\},$ and the underlying network $G$ is the whole graph including those dotted-line vertices and edges. Since we assume that the graph topology of $G_n$ is the only given information, i.e., two parameters $R_{\mathsf{i}}$ and $R_{\mathsf{s}}$ are unknown for all vertices in $G_n$. Hence, we shall assume that $R_{\mathsf{i}}=R_{\mathsf{s}}=1$ for all vertices and we have a simplified infection probability for vertex $v_a$

\begin{equation}
\label{eq:simp_spread_p}
    P(v_a \text{ is infected})=\frac{\sum\limits_{u_a} 1}{\sum\limits_{v\in \mathbf{S}}[\sum\limits_{u} 1 ]},
\end{equation}
which implies that the probability of $v_a$ being infected is proportional to the number of its infected neighbors. Note that, when $G$ is a tree network, this spreading model is equivalent to the one considered in \cite{who}. However, the probability defined in (\ref{eq:spread_p}) can be used to analyze more general cases such as graphs with cycles. In this paper, we aim to solve the epidemic source detection problem which is defined as follows:

\begin{equation}
\label{main_problem}
\begin{aligned}
& \underset{v\in G_n}{\text{maximize}} 
& & P(G_n\vert{v})\\
& \text{subject to}
& & G_n \subset G,
\end{aligned}
\end{equation}
where $P(G_n\vert{v})$ is the likelihood function assuming $v$ is the source, and $G$ is an almost $d$-regular graph with some irregular vertices, i.e., vertices with degree not equal to $d$. The maximum likelihood estimator for the epidemic source is the vertex $v$ with the maximum $P(G_n\vert{v})$ \cite{who}. In this paper, we assume that all irregular vertices have degree less than $d$.

\begin{definition}
For a given infection subgraph $G_n$ over the underlying graph $G$, $\hat{v}$ is an maximum likelihood estimator for the epidemic source in $G_n$, i.e., $P(G_n\vert \hat{v})=\max\limits_{v_i\in G_n}P(G_n\vert v_i).$
\end{definition}

In the remaining part of this section, we review the maximum likelihood estimation problem of this epidemic source in the simplest case: regular-tree networks. By Bayes' theorem, $P(G_n\vert{v})$ is the probability that $v$ is the {\it actual epidemic infection source} that leads to observing $G_n$. Now, let $\sigma_i$ be the possible spreading order starting from $v$, and let $M(v,G_n)$ be the collection of all $\sigma_i$ when  $v$ is the source in $G_n$. Then, we have
\begin{equation}
\label{pgv}
P(G_n\vert{v})=\displaystyle\sum_{\sigma_i \in {M(v,G_n)}}{P(\sigma_i\vert{v})}.
\end{equation}
In particular, for a $d$-regular tree, we have  \cite{who}:
\begin{equation}
\label{psigma}
P(\sigma_i \vert{v})=\displaystyle\prod_{k=1}^{n-1} \frac{1}{dk-2(k-1)}.
\end{equation}

Now, if the spreading has not reached the irregular vertices, then $P(\sigma_i \vert{v})=P(\sigma_j \vert{v})$ for all
$\sigma_i,\sigma_j \in M(v,G_n)$. By combining (\ref{pgv}) and (\ref{psigma}), we have
\begin{align*}
P(G_n\vert{v})&=\sum_{\sigma_i \in {M(v,G_n)}}{P(\sigma_i\vert{v})}
\\
&={\vert M(v,G_n) \vert}\cdot{P(\sigma\vert{v})}
\indent  \forall\sigma_i \in M(v,G_n)
\\
&={\vert M(v,G_n) \vert}\cdot\displaystyle\prod_{k=1}^{n-1} \frac{1}{dk-2(k-1)},
\end{align*}
which means that $P(G_n\vert v)$ is proportional to  $\vert M(v,G_n) \vert $. If we treat the rooted tree as a partially ordered set, then a spreading order is a linear extension corresponding to this poset. Hence, the quantity $\vert M(v,G_n) \vert$ is actually the number all linear extensions of the poset $G_n$ rooting at $v$. The authors in \cite{who} called $\vert M(v,G_n) \vert$  {\it rumor centrality}, which is crucial to solving the maximum likelihood estimation for degree-regular trees. We denote the vertex having the maximum $\vert M(v,G_n) \vert$ among all vertices in $G_n$ as $v_c$, that is

\begin{equation}
    v_c=\argmax\limits_{v\in G_n} \vert M(v,G_n)\vert .
\end{equation}

In particular, $v_c$ is called the {\it rumor center} when $G_n$ is a tree in \cite{who} and the authors in \cite{tan3,ma} established its equivalence to the graph distance center and tree centroid. When $G_n$ is a general graph, rumor center is only defined on the BFS spanning tree of $G_n$. To avoid confusion, in this paper we call the quantity $\vert M(v,G_n) \vert$ {\it epidemic centrality} and $v_c$ {\it epidemic center} for a general graph.

\begin{definition}
\label{def:distance}
For a graph $G$ and a vertex $v\in G$, the \textit{distance centrality} of $v$ can be defined as 
\begin{center}
    $distance~centrality$ of $v=\sum\limits_{u\in G}d(u,v)$,
\end{center}
where $d(u,v)$ is the shortest path distance, i.e., the number of edges along the shortest path from $u$ to $v$. The \textit{distance center} of a graph is the vertex with the smallest distance centrality. 
\end{definition}

\section{Trees with a Single degree-one Irregular Vertex}
\label{sec:finite}

\begin{figure*}
\begin{minipage}{0.85\linewidth}
\begin{minipage}{.55\textwidth}
\centering
\includegraphics[scale=0.25]{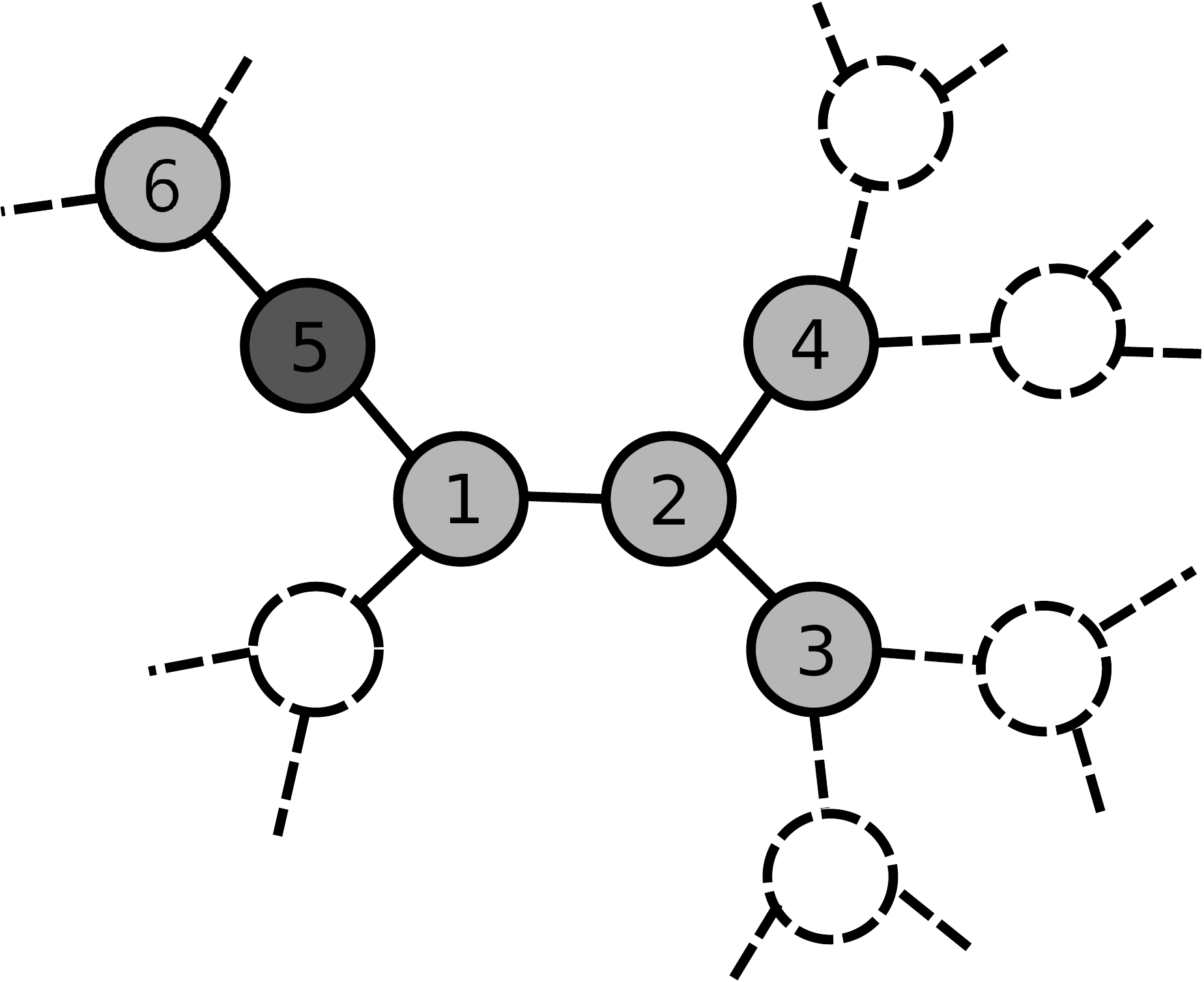}
\caption{Example of $G$ as a $3$-regular tree except $v_5$ and $G_n$ as a subtree with an irregular vertex $v_{ir}=v_5$. The maximum likelihood estimate $\hat{v}$ is $v_1$, moreover, the likelihood of $v_5$ is also greater than $v_2$. While a naive application of the rumor centrality in \cite{shah2}, i.e., the rumor center $v_c$ of $G_n$, yields both $v_1$ and $v_2$.}
\label{fig:irregular_ex}
\end{minipage}
\makeatletter\def\@captype{table}\makeatother
\begin{minipage}{.55\textwidth}
\caption{Numerical Example of $P(G_6\vert v_i)$ using $G_6$ in Fig. \ref{fig:irregular_ex}}
\renewcommand\arraystretch{1.2}
\label{tab:single_example}
\begin{tabular}{|c|c|c|c|c|c|c|} 
\cline{2-6}
\multicolumn{1}{c|}{}& \multicolumn{5}{c|}{$m^{v_5}_{v}(G_6,k)$} & \multicolumn{1}{|c}{}
\\
\hline
$v$ & $k=1$ & $k=2$ & $k=3$ & $k=4$ & $k=5$ & $P( G_6 \vert v)$ \\
\hline
$v_1$ & $0$ & $8$ & $6$ & $6$ & $0$ & 0.0149\\
\hline 
$v_5$ & $10$ & $0$ & $0$ & $0$ & $0$ & 0.0138 \\ 
\hline
$v_2$ & $0$ & $0$ & $6$ & $8$ & $6$ & 0.0114 \\
\hline 
$v_3$ & $0$ & $0$ & $0$ & $2$ & $2$ & 0.002 \\ 
\hline 
$v_4$ & $0$ & $0$ & $0$ & $2$ & $2$ & 0.002 \\ 
\hline 
$v_6$ & $0$ & $2$ & $0$ & $0$ & $0$ & 0.0018 \\ 
\hline 
\end{tabular}
\end{minipage}
\end{minipage}
\end{figure*}

In this section we study the effect on the maximum likelihood due to an {\it irregular vertex} in a regular tree network. Assuming that the original underlying network is a regular tree, this irregular vertex is a vertex with different number of neighborhoods compared to other vertices. For example, in a regular tree $G$ of bounded size, if the infected subgraph $G_n$ contains a leaf $v$ of $G$, then $v$ is an irregular vertex in $G$ since it has no other neighbors except its parent vertex. This models a person under quarantine in a contact tracing network with fewer neighborhoods than other people. 

We consider the case where the degree of the irregular vertex is less than that of all other vertices. Note that a degree-one irregular vertex is a leaf of both $G$ and $G_n$. Now, assume that the infection subgraph $G_n$ only has a single irregular vertex denoted as $v_{ir}$, where $deg(v_{ir})<d$ in $G$. Take Fig. \ref{fig:irregular_ex} for example, we have $v_{ir}=v_5$ since $deg(v_{ir})=2$ and other vertices are of degree $3$ in $G$. In the following, we study how $v_{ir}$ affects the maximum-likelihood estimation performance.

In particular, we compare this single irregular vertex special case with a naive prediction that assumes an underlying {\it infinite graph}. This illustrates that ignoring the irregular effect in the finite graph ultimately leads to a wrong estimate and thus requires an in-depth analysis and new epidemic source detection algorithm design for the general case of {\it finite graphs}.

\subsection{Impact of the Irregularity On $P(G_n | v)$}

\begin{example}
\label{exam:1}
Let $G$ be a infinite $3$-regular tree except $v_5$ has a different degree, and $G_6=$ is a subgraph of $G$ shown in Fig. \ref{fig:irregular_ex}. Consider $P(G_6 \vert v_1)$ and with a spreading order $\sigma: v_1\rightarrow v_2 \rightarrow v_5 \rightarrow v_3 \rightarrow v_4 \rightarrow v_6 $, we have $P(\sigma\vert v_1)=(1/3)\cdot (1/4)\cdot (1/4)\cdot (1/5)\cdot (1/6)$. Had $v_5$ not been the irregular vertex, then $P(\sigma\vert v_1)=(1/3)\cdot (1/4)\cdot (1/5)\cdot (1/6)\cdot (1/7)$. This demonstrates that the order at which the virus spreads to the irregular vertex $v_5$ is important when computing  $P(\sigma \vert v_1)$. In particular, $P(G_6\vert v_1)\approx 0.0149$. We also have $P(G_5\vert v_2)\approx0.0114$. Now, observe that $v_1$ and $v_2$ are two vertices with the largest epidemic centrality among all vertices in $G_6$, but $P(G_6\vert v_1)>P(G_6\vert v_2)$, and thus $\hat{v}=v$. Note that, the likelihood of $v_5$ being the source is also greater than that of $v_2$ even $v_2$ has a greater epidemic centrality.
\end{example}

Example \ref{exam:1} reveals some interesting properties of irregular effects due to even a single irregular vertex:  
\begin{itemize}
\item $P(\sigma_i \vert v)$ increases with how soon the  irregular vertex appears in $\sigma_i$ (as ordered from left to right of $\sigma_i$). 
\item When there is at least one irregular vertex in $G_n$, then $P(G_n\vert v)$ is no longer proportional to $\vert M(v,G_n)\vert$.
\end{itemize}

This means that $P(\sigma_i\vert v)$ is no longer a constant for each $i$, and depends on the position of the {\it irregular vertex} in each spreading order. We proceed to compute $P(\sigma_i\vert v)$ as follows. For brevity of notation, let $v_{ir}$ be the irregular vertex and let
\begin{align*}
M^{v_{ir}}_{v}(G_n,k)&= \lbrace \sigma\in M(v,G_n)\vert  v_{ir}\mbox{ is the }k\mbox{th} \mbox{ element of }\sigma\rbrace ;
\\
P^{v_{ir}}_{v}(G_n,k)&=P(\sigma \vert v),\mbox{ for }\sigma \in M^{v_{ir}}_{v}(G_n,k),
\end{align*}
where $M^{v_{ir}}_{v}(G_n,k)$ is the set of all the spreading orders  starting from $v$ and with $v_{ir}$ at the $k$th position, and its size is the combinatorial object of interest:
\begin{equation} 
\label{minimvgn}
m^{v_{ir}}_{v}(G_n,k)  = \vert M^{v_{ir}}_{v}(G_n,k)\vert.
\end{equation}
Let $D$ be the distance (in terms of number of hops) from $v$ to $v_{ir}$. Then we have 
\begin{equation}
\label{eq:mvgn}
\vert M(v,G_n)\vert=\sum\limits^{n-t^v_{v_{ir}}+1}_{k=D+1} m^{v_{ir}}_{v}(G_n,k).
\end{equation}
Now, (\ref{eq:mvgn}) shows that $M(v,G_n)$ can be decomposed into $M^{v_{ir}}_{v}(G_n,k)$ for $k=D+1,D+2,\dots ,n-t^v_{v_{ir}}+1$. This  decomposition allows us to handle the irregular effect due to the different position of the irregular vertex in each spreading order. For example, in Table \ref{tab:single_example} we list out all $m^{v_{ir}}_{v}(G_n,k)$ for each possible starting vertex $v$ and the position $k$ of $v_{ir}$. Let $P^{v_{ir}}_{v}(G_n,k)$ be the corresponding probability for each $k$. We can rewrite $P(G_n\vert v)$ for the case with single irregularity as: 
\begin{equation}
\label{finiteP}
P(G_n\vert v)=\sum\limits^{n-t^v_{v_{ir}}+1}_{k=D+1} m^{v_{ir}}_{v}(G_n,k)\cdot P^{v_{ir}}_{v}(G_n,k).
\end{equation} 
Thus, solving (\ref{main_problem}) means finding the vertex $\hat{v}$ that solves 
\begin{equation}
P(G_n\vert \hat{v})=\max\limits_{v_i\in G_n}P(G_n\vert v_i).
\end{equation}
Since $P(G_n\vert v)$ is no longer proportional to $\vert M(v,G_n)\vert$, we now describe how to compute $P(G_n\vert v)$ in $G_n$ over an underlying $d$-regular tree except the degree of $v_{ir}$ is $d'$. First, consider $P^{v_{ir}}_{v}(G_n,k)$ and let $z_d(i)=(i-1)(d-2)$, then
\begin{align}
\label{newp}
P^{v_{ir}}_{v}(G_n,k)=\prod\limits^{k-1}_{i=1}\frac{1}{d+z_d(i)}\cdot\prod\limits^{n-2}_{i=k-1}\frac{1}{d+z_d(i)+(d'-2)},
\end{align}
where the first factor of $P^{v_{ir}}_{v}(G_n,k)$ in (\ref{newp}) is the probability that $k$ vertices are infected once the virus reaches the irregular vertex, i.e.,  $v_{ir}$ is the $k$th vertex infected in $G_n$, and the second factor is the probability that all remaining $n-k$ vertices are infected thereafter. On the other hand, the value of $m^{v_{ir}}_{v}(G_n,k)$ in (\ref{minimvgn}) is dependent on the network topology, and thus there is no closed-form expression in  general (though when $G_n$ is a line, a closed-form expression for $m^{v_{ir}}_v(G_n,k)$ is given in (\ref{eq:line_mvgn})). We now use a special case, line graph, to demonstrate how an irregular vertex and network topology affect the probability $P(G_n\vert v)$.

\subsection{Analytical Characterization of Likelihood Function}

Suppose $G$ is a finite degree-regular tree and $G_n$ is a line graph with a single irregular vertex due to the bounded size of $G$. 
Without loss of generality, suppose $n$ is odd (to ensure a unique $v_c$) and $n=2t+1$ for some $t$. We label all the vertices in $G_n$ from $1$ to $2t+1$ and assume that $v_{2t+1}$ is the end vertex, i.e.,irregular vertex with degree $1$. To compute $P(G_{n}\vert v_i)$ for $v_i\in G_n$, from (\ref{finiteP}) and (\ref{newp}), we already have $P^{v_{ir}}_{v_i}(G_n,k)$, so we need to compute $m^{v_{ir}}_{v_i}(G_n,k)$. The enumeration of $m^{v_{ir}}_{v_i}(G_n,k)$ can be accomplished in polynomial-time complexity with a path-counting message-passing algorithm (see, e.g., Chapter 16 in \cite{mackay}). In particular, we have a closed-form expression for $m^{v_{ir}}_{v}(G_n,k)$ given by: 
\begin{equation}
\label{eq:line_mvgn}
m^{v_{ir}}_{v_i}(G_n,k)=\binom{k-2}{k-n+i-1},
\end{equation}  
when $i\neq n$, leading to an analytical formula for $P(G_{n}\vert v_i)$: 

$P(G_{n}\vert v_i)=$
\begin{equation}
\label{firsteqn}
\begin{cases}
\displaystyle\prod_{l=1}^{n-1}\frac{1}{z_d(l)+1},&\mbox{$i=n$;}
\\
\displaystyle\sum_{k=n-i+1}^{n}\binom{k-2}{k-n+i-1}\cdot P^{v_{ir}}_{v_i}(G_n,k),&\mbox{otherwise,}
\end{cases}
\end{equation}
where $P^{v_{ir}}_{v_i}(G_n,k)$ is given in (\ref{newp}).

In (\ref{firsteqn}), we suppose that $n$ is odd. Using (\ref{firsteqn}), let us numerically compute $P(G_n\vert v_i)$ for all $v_i$ in Fig. \ref{fig:line_d4}, where $G$ is a $4$-regular tree and $G_n$ is a line graph with a single irregular vertex $v_{ir}=v_{n}$ as boundary for different values of $n=7,8,9,10$. The $x$-axis is the vertex $v_i$ where $i=1,2,\dots,10$, and the $y$-axis plots $P(v_i=v^{\star} \vert G_n)$. Fig. \ref{fig:line_d4} illustrates that the influence due to the irregular vertex on $P(v_i=v^{\star} \vert G_n)$ dominates that of the epidemic center when $n=7,8,9$. However, the situation reverses when $n=10$ (i.e., the epidemic center on $P(G_n\vert v_i)$ is dominant thereafter). 

\begin{theorem}
\label{thm:longterm}
Suppose $G$ is a $d$-regular graph ($d>2$) with finite order. If $G_n$ is a line-graph with a single irregular vertex at one end of the line graph, then  there is a constant $j$ such that $P(G_n\vert v_c)>P(G_n\vert v_{ir}) $ when $n>j$.
\end{theorem}
Remark: When $n$ increases, i.e., the distance between $v_c$ and $v_{ir}$ increases, then the location of $\hat{v}$ in $G_n$ converges to the neighborhood of the epidemic center. 
\begin{example}
To verify Theorem \ref{thm:longterm}, we plot $P(G_n\vert v_i)$ for a line graph $G_n$ with $G$ being a finite $4$-regular graph in Fig. \ref{fig:line_d4}. Clearly, we have $j=9$.
\end{example}

Theorem \ref{thm:longterm} implies that, for any $d$-regular underlying graph, when $G_n$ is a line graph with a single irregular vertex, the influence of the irregular vertex $v_{ir}$ on $P(v_i|G_n)$ decreases monotonically as $n$ grows. In fact, this reduces to the special case in \cite{shah2}, when $n$ goes to infinity asymptotically, i.e., $\hat{v}$ is the ML estimation of the source.

\begin{figure}
\begin{center}
\includegraphics[scale=0.4]{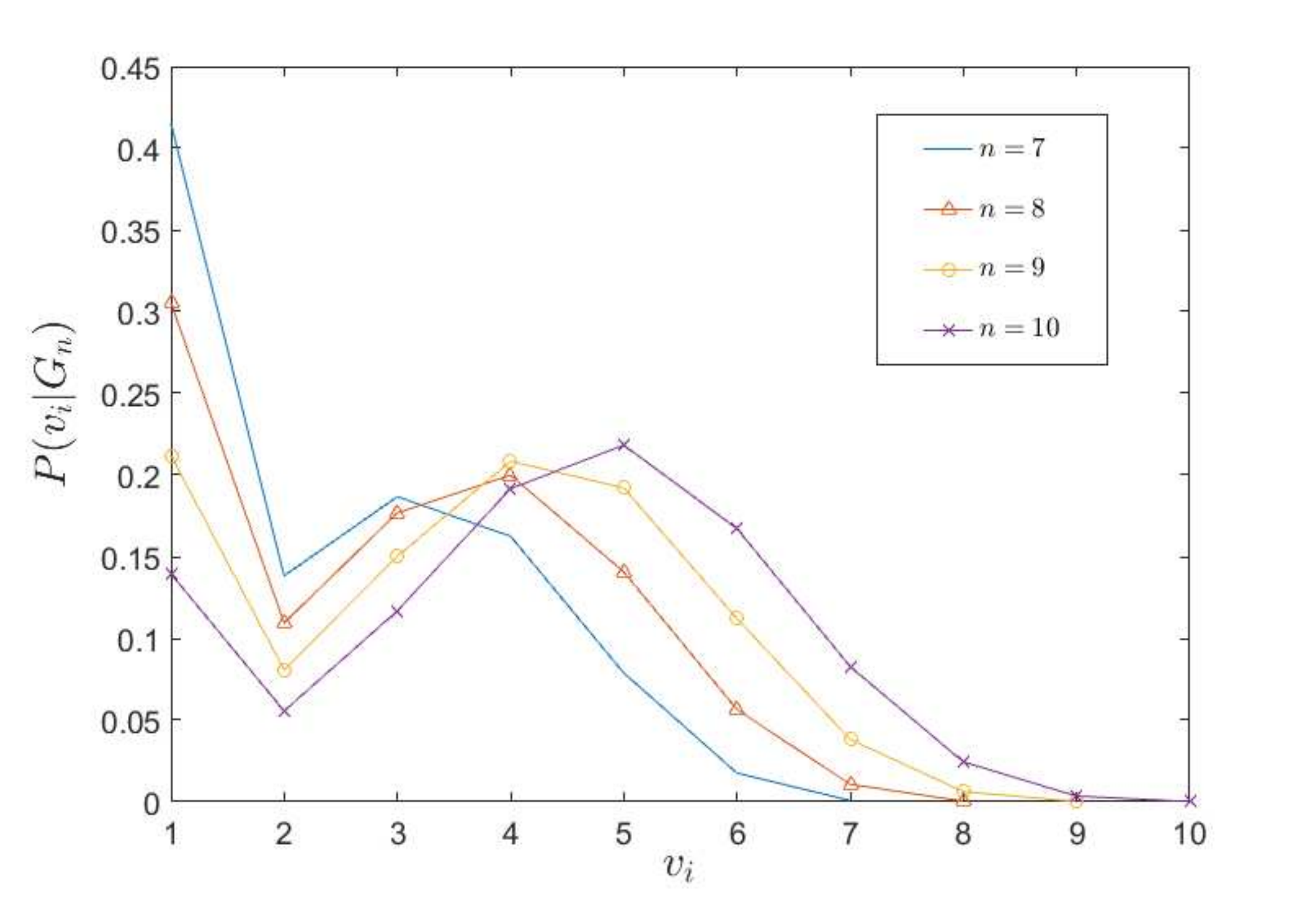}
\caption{$P(G_n\vert v)$, where $G_n$ is a line graph with a single irregular vertex $v_1$ over an underlying $4$-regular finite graph.}
\label{fig:line_d4}
\end{center}
\end{figure}

\subsection{Optimality Characterization of Likelihood Estimate} 
\label{subsec:location}

Example \ref{exam:1} reveals that in addition to the spreading order, the distance (number of hops) between the irregular vertex and $v$ also affects the likelihood probability $P(G_n\vert v)$. We can conclude that there are two factors affect the likelihood of $v$ being the source, one is the distance $d(v, v_{ir})$ and the other one is the epidemic centrality $\vert M(v,G_n)\vert$. Intuitively, if we consider two vertices in $G_n$ say $v_a$ and $v_b$ where $\vert M(v_a, G_n)\vert > \vert M(v_b, G_n)\vert$ and $d(v_a, v_{ir}) < d(v_b, v_{ir})$, then the above observation may lead us to $P(G_n\vert v_a)>P(G_n\vert v_b)$. We formalize this optimality result that characterizes the probabilistic inference performance between any two vertices and the location of $\hat{v}$ in $G_n$ with a single irregular vertex $v_{ir}$ in the next theorem.

\begin{lemma}
\label{lem:key_lem}
Let $G$ be a tree of size $n$, for any three vertices say $v_a$, $v_b$ and $v_{ir}$, if one of the following conditions is satisfied
\begin{enumerate}
    \item $\vert M(v_a, G_n)\vert \geq \vert M(v_b, G_n)\vert$ and $d(v_a, v_{ir}) < d(v_b, v_{ir})$
    \item $\vert M(v_a, G_n)\vert > \vert M(v_b, G_n)\vert$ and $d(v_a, v_{ir}) \leq d(v_b, v_{ir})$,
\end{enumerate}
 then we have 
\begin{center}
    $\sum\limits_{i=d(v_a, v_{ir})+1}^{k}m_{v_a}^{v_{ir}}(G_n, i)\geq \sum\limits_{i=d(v_a, v_{ir})+1}^{k}m_{v_b}^{v_{ir}}(G_n, i)$,
\end{center}
for all possible $k$.
\end{lemma}
Remark: Lemma \ref{lem:key_lem} applies for any positive degree of $v_{ir}$. We can verify Lemma \ref{lem:key_lem} by Fig. \ref{fig:irregular_ex}. Let $v_a=v_1$, $v_b=v_2$ and $v_{ir}=v_5$, then $\vert M(v_1, G_n)\vert = \vert M(v_2, G_n)\vert=20$ and $d(v_1, v_5) < d(v_2, v_5)$ which satisfies the first condition of Lemma \ref{lem:key_lem}. Hence, the partial sum of $\vert M(v_1, G_6) \vert$ is always greater or equal to the partial sum of $\vert M(v_2, G_6) \vert$.  

\begin{theorem}
\label{thm:keyfinite}
Let $G$ be a $d-$regular tree with a single irregular vertex and $G_n \subseteq G$ be a subtree of $G$ with a single irregular vertex $v_{ir} \in G_n$. where $deg(v_{ir})<d$. Then, the maximum likelihood estimator $\hat{v}$ with maximum probability $P(G_n\vert v)$ is located on the path from the $v_c$ to $v_{ir}$. 
\end{theorem}

The above theorem can be immediately deduced from Lemma \ref{lem:key_lem}. In addition,  we can leverage Lemma \ref{lem:key_lem} in the case when $deg(v_{ir})>d$ to reduce the search space of the optimization problem (\ref{main_problem}). We can conclude only that the ML estimator is on the path from  $v_c$ to $v_{ir}$, which is illustrated in Fig. \ref{fig:thm1}, thus narrowing down the search for the ML estimator to a vertex on this path.


\begin{figure}
\begin{center}
\includegraphics[scale=0.23]{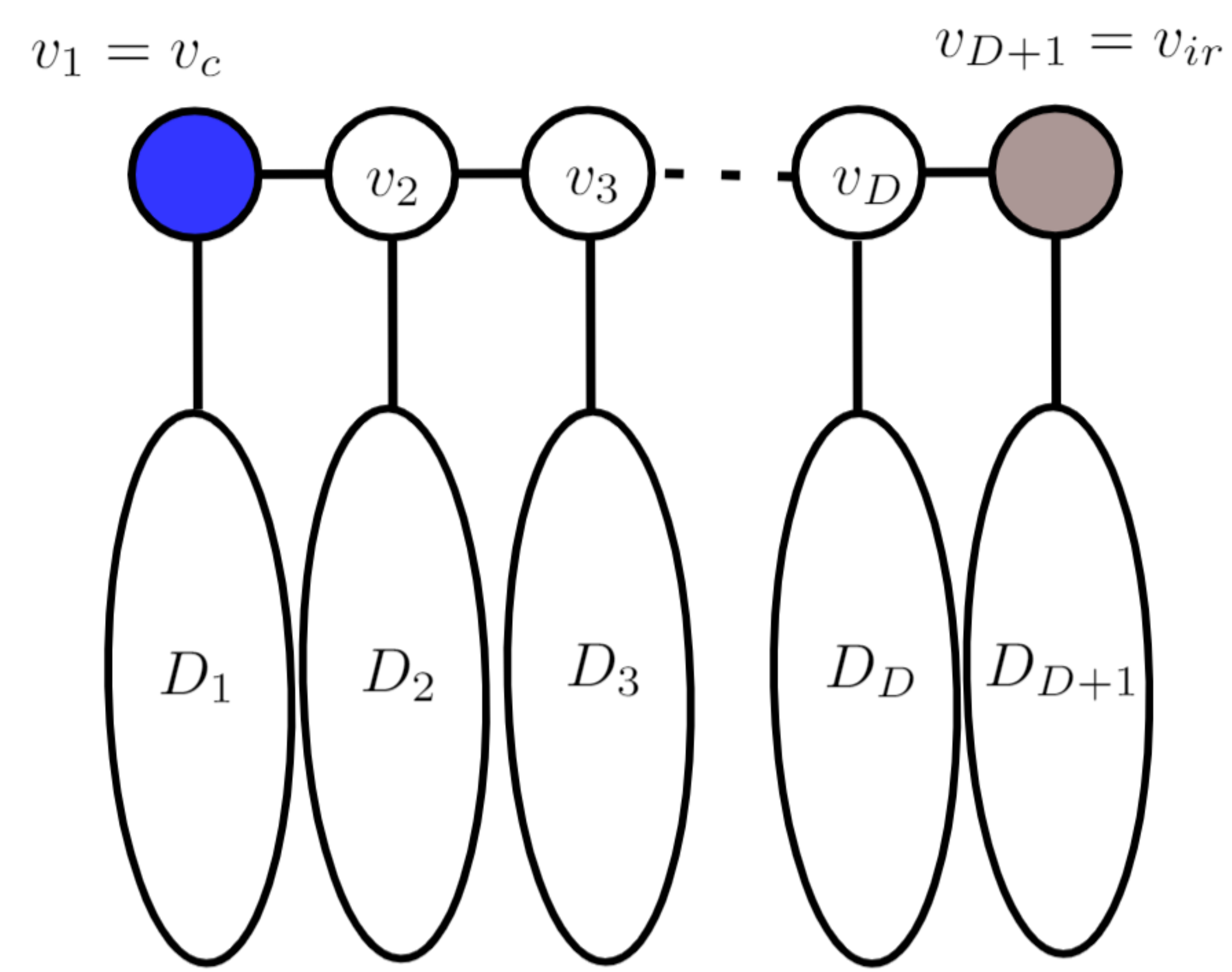}
 \caption{An intuitive illustration of Theorem \ref{thm:keyfinite}. If $v_1=v_c$ is the epidemic center and $v_{D+1}=v_{ir}$ is the irregular vertex where $deg(v_{ir})$ less than other vertices.  
   From Lemma \ref{lem:key_lem}, we have $P(G_n\vert v_i)> P(G_n\vert v)$ for all $v \in D_i$, where $i=1,2,\ldots, D+1$. Hence, we only need to consider $v_1, v_2, \ldots, v_{D+1},$ to find the ML estimator for the source.}
\label{fig:thm1}
\end{center}
\end{figure}

\section{Trees with Multiple Degree-one Irregular Vertices}
In this section, we consider the case when $G_n$ has  more than a one irregular vertex (naturally, this also means $d>2$ in $G$ ruling out the trivial case of $G$ being a line). The key insight from the single irregular vertex analysis still holds: Once the virus reaches an irregular vertex in $G$, $\hat{v}$ can be located near this very first  infected irregular vertex. In addition, the algorithm design approach is to decompose the graph into subtrees to narrow the search for the maximum-likelihood estimate solution. To better understand the difficulty of solving the general case, we start with a special case: The entire finite underlying network is infected, i.e., $G_n=G$, then $P(G_n \vert v)=1/n$ for each vertex in $G_n$, as each vertex is equally likely to spread the virus to all the other vertices in $G$ to yield $G_n=G$. In this case, $P(G_n\vert \hat{v})$ is exactly the minimum detection probability. So the bound of $P(G_n\vert \hat{v})$ given in previous study are not suitable for the case with irregular vertices. Therefore, when simulating the virus spreading in a network, we will set an upper bound $n/k$ of the number of irregular vertices where $k$ is some integer greater than 1, once the number of irregular vertices in $G_n$ reaches to $\lceil n/k \rceil$, then we will stop the spreading process.   
\subsection{Degree-Regular Tree ($d\geq 3$) Special Case: $G_n$ is Broom-Shaped} 
\label{broom_sec}

\begin{figure}
\begin{center}
\includegraphics[scale=0.3]{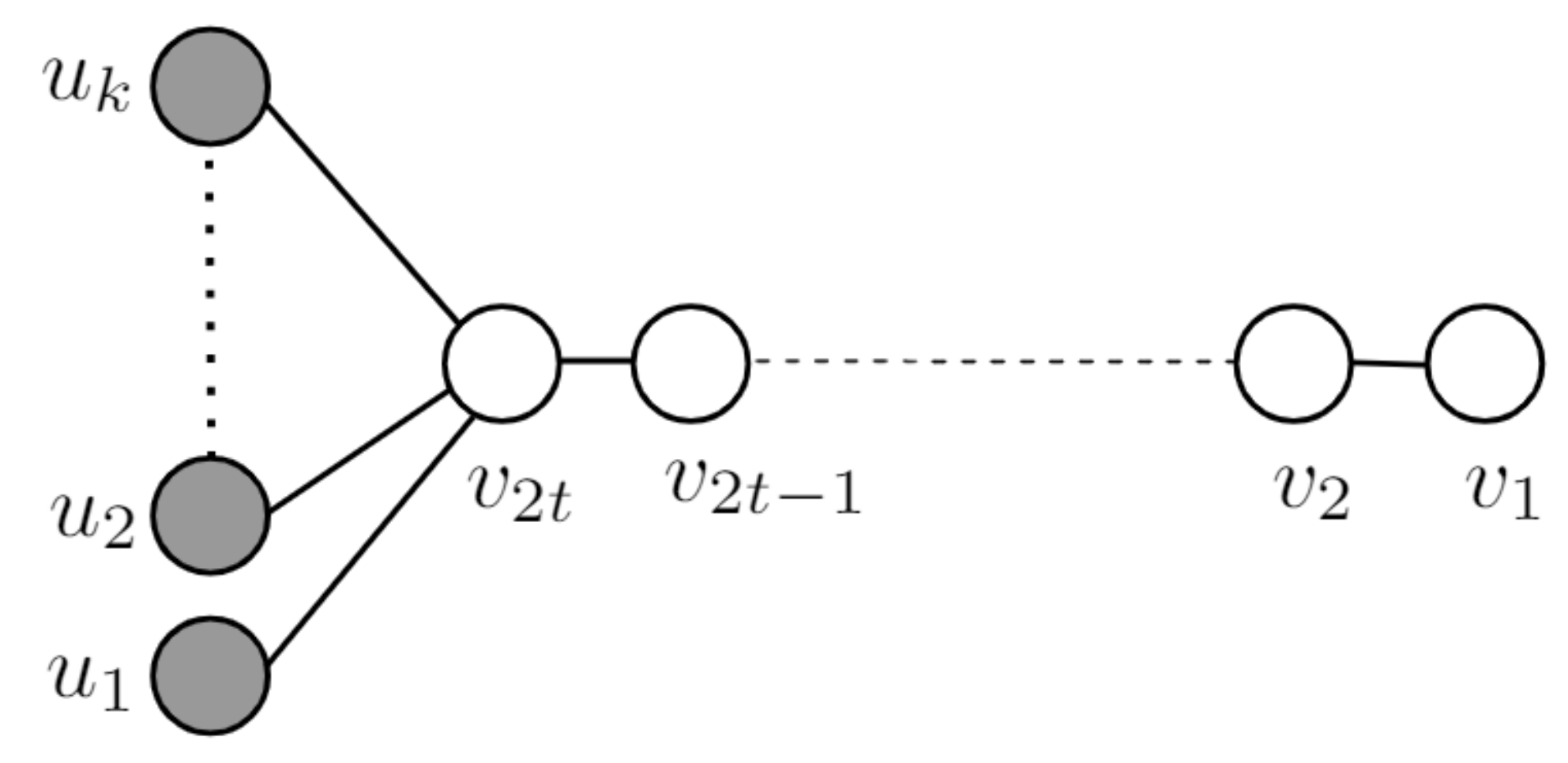}
  \caption{$G_n$ as a broom graph with $k$ star-like end vertices $e_1$ to $e_k$.}
  \label{fig:broom}
\end{center}
\end{figure}

In Section \ref{sec:finite}, we have shown that, when $G_n$ is sufficiently large, the effect of the single irregular vertex on $P(v\vert G_n)$ for each vertex $v$ on the line graph $G_n$ is dominated by $v_c$. Now, we study the effect of multiple degree-one irregular vertices on a class of graphs whose topology is richer than the line graph in Section \ref{sec:finite}. In particular, as shown in Fig. \ref{fig:broom}, we add degree-one irregular vertices to $v_{2t}$, so that when $G$ is $d$-regular, then there will be at most $d-1$ degree-one irregular vertices in $G_n$. We call this the {\it broom} graph. We can compute $P(v\vert G_n)$ by extending the result in Section \ref{sec:finite}. Let $P_{v_i}^{\lbrace u_1,u_2,\dots,u_k\rbrace}(\lbrace h_1,h_2,\dots ,h_k\rbrace, G_n)$ be the probability of the spreading order  starting from $v_i$ with the irregular vertex set $\lbrace u_1,u_2,\dots ,u_k\rbrace$ and their position set $\lbrace h_1,h_2,\dots ,h_k\rbrace$ in this spreading order. We do not assume that $h_i$ is the position of $u_i$, as it can be the position of any irregular vertex in $G_n$. The probability $P_{v_i}^{\lbrace u_1,u_2,\dots,u_k\rbrace}(\lbrace h_1,h_2,\dots ,h_k\rbrace, G_n)$ can be obtained by the same analysis in (\ref{newp}). To compute $m_{v_c}^{\lbrace u_1,u_2,\dots,u_k\rbrace}(\lbrace h_1,h_2,\dots,h_k\rbrace, G_n)$, we first consider the line-shaped part of $G_n$, i.e., the part $\lbrace v_1,v_2,\dots,v_{2t}\rbrace$, say $G_n'$. From the previous discussion, we have 
\begin{center}
$m^{v_{2t}}_{v_i}(G_n',j)=\binom{j+(2t+i-1)}{j}$, 
\end{center}
and for each spreading order that $v_{2t}$ lies on the $j$th position, the degree-one irregular vertices $u_1,u_2,\dots,u_k$ can be placed on any position after the $j$th position. So for each spreading order in $m^{v_{2t}}_{v_i}(G_n',j)$, there are $k!\cdot\binom{n-k-j+1}{k}$ corresponding spreading orders in $G_n$. Thus, we have 
\begin{equation}
\label{p_multi}
m_{v_i}^{\lbrace u_1,\dots,u_k \rbrace}(G_n,\lbrace h_1,\dots,h_k\rbrace )=k!\sum\limits^{h_1-1}_{j=2t-i+1}\binom{j-2}{2t-i-1}.
\end{equation}

With $P_{v_i}^{\lbrace u_1,u_2,\dots,u_k\rbrace}$ and $m_{v_i}^{\lbrace u_1,u_2,\dots,u_k\rbrace}$, we can now compute the probability $P(v_i\vert G_n)$ by going through all possible $\lbrace h_1,h_2,\dots,h_k\rbrace$. Fig. \ref{fig:d6n25} shows that even though there are five irregular vertices, the effect of $v_c$ on $P(v\vert G_n)$ eventually dominates that of the irregular vertices as $n$ grows from 37 to 39. These result implies that: When there are more degree-one irregular vertices in $G_n$, $n$ needs to be sufficiently large to offset the effect of irregular vertices, i.e., for the transition phenomenon to take place. 
For other $d$ and $n$ in the $broom$ graph, as shown in the proof of Theorem \ref{thm:longterm}, we can prove this in the same way to conclude that, if we fix the number of irregular vertex, the probability $P(v_c\vert G_n)$ will be  greater than $P(v_{ir}\vert G_n)$ when $n$ is large enough. 

\begin{figure}
\begin{center}
\includegraphics[scale=0.4]{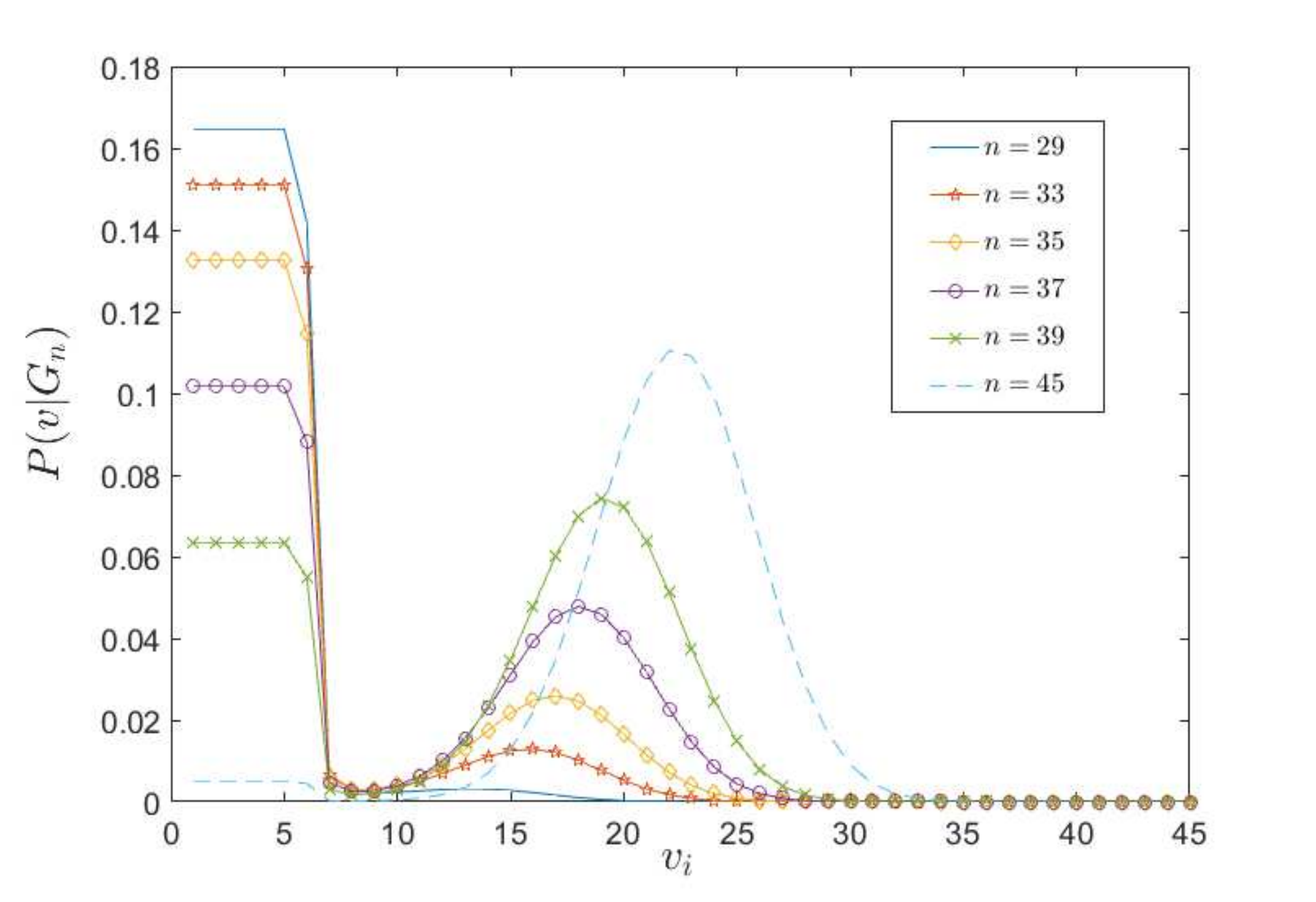}
\caption{Probability distribution of each vertex on $G_n$ with five degree-one irregular vertices when $G$ is $6$-regular, The $y$-axis plots the probability $P(v_i\vert G_n)$ and the $x$-axis plots the vertex $v_i$'s number $i$. In particular, $v_1,\dots,v_5$ are the leaves (degree-one irregular vertices) corresponding to $u_1,\dots,u_k$ in Fig. \ref{fig:broom}, where $k=5$. Observe that the transition phenomenon happens when $n$ grows from $37$ to $39$.}
\label{fig:d6n25}
\end{center}
\end{figure}

\subsection{Message-passing Algorithm}

We propose a message-passing algorithm to find  $\hat{v}$ on the finite regular tree $G$ by leveraging the key insights derived in the previous sections. We summarize these features as follows:
\begin{enumerate}
\item If there is only a single irregular vertex $v_{ir}$ in $G_n$, then $\hat{v}$ is located on the path from $v_c$ to $v_{ir}$.
\item If $G_n=G$, then for all $v_i\in G_n$, $P(G_n\vert v_i)=1/n$.
\item If $G_n$ has $q$ degree-one irregular vertices, then there exists an $n'$ such that, if $n>n'$, then $P(G_n\vert v_c)>\max\limits_{1\leq i \leq k}\lbrace P(G_n\vert v_{u_i})\rbrace$. Furthermore, $n'$ increases as $q$ increases.
\item If two vertices $v_1$ and $v_2$ are on the symmetric position of $G_n$, then $P(G_n\vert v_1 )=P(G_n\vert v_2)$. For example, $u_1, u_2, \ldots , u_k$ are topologically symmetric in Fig. \ref{fig:broom}. 
\end{enumerate}
In particular, Feature 1 is the optimality result related to decomposing $G_n$ into subtrees to search for $\hat{v}$. The subtree $t_{ML}$ in $G_n$ corresponds to first finding the decomposed subtree containing $v_c$ and the likelihood estimate needed for Theorem \ref{thm:keyfinite} to apply. Then, Features 3 and 4 identify $\hat{v}$ on a subtree $t_{ML}$ of $G_n$ as Theorem \ref{thm:keyfinite} only pinpoints the relative position of $\hat{v}$.

\begin{algorithm}[htb]
\caption{Message-passing algorithm to compute $\hat{v}$ for $G_n$ with multiple irregular vertices}
\label{algo:multi_algo}
\begin{algorithmic}
\Require 
$G_n$, $\kappa=\lbrace \rbrace$ 
\State Step 1: Compute $v_c$ of $G_n$.
\State Step 2: Choose $v_c$ as the root of a tree and  use a message-passing algorithm to count the number of irregular vertices on each branch of this rooted tree.
\State Step 3: Starting from $v_c$, and at each hop choose the child with the maximum number of irregular vertices (if there were more children with the same maximal number of irregular vertices, then choose all of them). This tree traversal yields a subtree $t_{ML}$ rooted at $v_c$. 
\State \textbf{Output:} $\kappa=\lbrace \mbox{parent vertices of leaves of } t_{ML},v_c\rbrace$
\end{algorithmic}
\end{algorithm}

Algorithm \ref{algo:multi_algo} first finds the epidemic center of $G_n$ and then determines the number of irregular vertices corresponding to each branch of the epidemic center $v_c$. The final step is to collect vertices on the subtree where $\hat{v}$ is, leading to a subtree of $G_n$ denoted as $t_{ML}$. Observe that each step requires $O(n)$ computational time complexity and $t_{ML}$ in a graph with multiple irregular vertices is akin to the \emph{path} from $v_c$ to the irregular vertex in an infected subgraph with a single irregular vertex in Section \ref{sec:finite}. Finally, we obtain a set $\kappa$ containing the parent vertices of the leaves of $t_{ML}$ and $v_c$. 

\begin{figure}
\begin{center}
\includegraphics[scale=0.26]{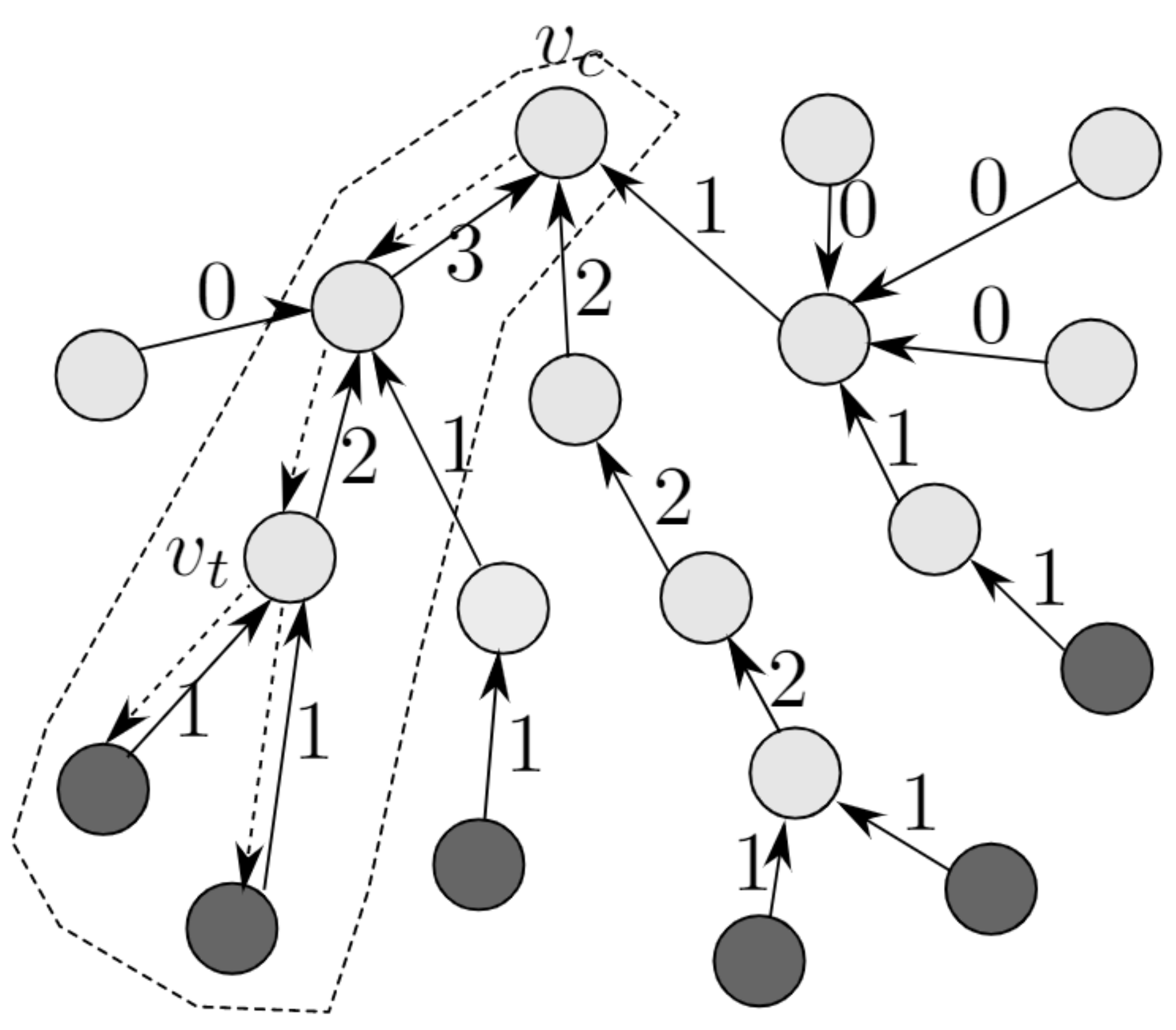}
    \caption{An illustration of how Algorithm \ref{algo:multi_algo} works on a tree graph rooted at $v_c$ with six degree-one irregular vertices (more shaded). Observed that $v_c$ branches out to three subtrees. Here, $t_{ML}$ is the subtree containing the five vertices within the dotted line. The numerical value on the edge indicates the message containing the number of end vertices being counted.}
  \label{algo2graph}
\end{center}
\end{figure}

Next, we use the example in Fig. \ref{algo2graph} to illustrate Algorithm \ref{algo:multi_algo}. Let $G_{19}$ be the network in Fig. \ref{algo2graph} with the six degree-one irregular vertices depicted as more shaded. Suppose $v_c$ is determined by the end of Step 1. Then, Step 2 enumerates the number of irregular vertices at each branch of the subtrees connected to $v_c$, and these numbers are then passed iteratively from the leaves to $v_c$. These messages correspond to the numerical value on the edges in Fig. \ref{algo2graph}. The message in Step 2 is an $upward$ (leaf-to-root) message. Step 3 is a message passing procedure from $v_c$ back to the leaves, which is a $downward$ message, and the message is the maximum of number of irregular vertices in each branch. For example, the message from $v_c$ to $child(v_c)$ is $\max\lbrace 1,2,3\rbrace$ which is $3$. Lastly, the second part of Step 3 collects those vertices whose $upward~message=downward~message$. For example, the left hand-side child of $v_c$ is first added to $t_{ML}$, and then $v_t$ is added to $t_{ML}$, and finally, the two leaves on the left hand-side is added to $t_{ML}$. Observe that $t_{ML}$ must be connected.

\subsection{Simulation Results for Finite General Tree Networks}

We simulate the virus spreading on finite regular tree networks and general tree networks with $\vert G\vert =1000$ and $\vert G_n\vert=100$. Each vertex in a general tree network is not larger than a positive integer $d_m$. The construction of $G$ is a random branching process in which we start with a single vertex $v_1$, and then randomly pick an integer, say $i$, from $1$ to $d_m-1$ to be the number of children of $v_1$, and then to assign $v_2$ to $v_{i+1}$ to be the neighborhood of $v_1$. Recursively applying these steps generates a finite tree $G$ with five thousand vertices whose maximum degree is not larger than $d_m$.

We simulate a thousand times the spread of the virus on $G$ by picking $v^{\star}$, the true source, uniformly on $G$, and compare the average performance of Algorithm \ref{algo:multi_algo}, a naive heuristic that simply uses the rumor centrality ($\mathsf{RC}$) \cite{who}, Jordan centrality ($\mathsf{JC}$) \cite{ZhuLei}, and a spectral based method called Dynamical Age ($\mathsf{DA}$) \cite{spectral_rumor}. To fairly compare these algorithms,  when Algorithm \ref{algo:multi_algo} yields a set with $\vert \kappa \vert$ vertices, then other methods find a set of $\vert \kappa \vert$ vertices having the top $\vert \kappa \vert$ maximum score for all $v$ of $G_n$. Obviously, the size of the solution set $\vert \kappa \vert$ depends on the topology of $G_n$ in each run of the simulation, and thus is not a constant in general over that thousand times. To quantify the performance of these two algorithms, let us define the error function of a vertex set $\eta$:
\begin{center}
$\mbox{error}(\eta)=min \lbrace d(v,v^{\star})\vert \forall v\in \eta\rbrace$.
\end{center}
This is the smallest number of hops between $v^{\star}$ and the nearest vertex in the set $\eta$. As shown in Table \ref{tab:compare_general}, we can observe that the number of vertices in $\kappa$ is surprisingly small, moreover, $\vert\kappa\vert$ is decreasing as $d_m$ increases. 

\begin{table}
\caption{Average error (in terms of number of hops) comparing Algorithm \ref{algo:multi_algo} to Top-$k$ version of Rumor Centrality ($\mathsf{RC}$), Jordan Center ($\mathsf{JC}$) and Dynamical Age ($\mathsf{DA}$), when $G$ is a $d$-regular graph, for $d=3,4,5,6$.}
\label{tab:compare}
\centering
\begin{tabular}{c|c|c|c|c|c}

 $d_m$& $\vert\kappa\vert$& Algo. \ref{algo:multi_algo} & $\mathsf{RC}$  & $\mathsf{JC}$ & $\mathsf{DA}$
\\
\hline
3 & 6.34 & \textbf{1.44} & 3.26 & 4.55 & 2.93
\\
\hline
4 &5.65& \textbf{1.50} & 2.64 & 3.52 & 2.18
\\
\hline
5 &4.05 & \textbf{1.48} & 2.36 & 3.18 & 2.05
\\
\hline
6 & 3.72& \textbf{1.40} & 2.32 & 2.94 & 2.00
\end{tabular}
\end{table} 

\begin{table}
\caption{Average error (in terms of number of hops) comparing Algorithm \ref{algo:multi_algo} to Top-$k$ version of Rumor Centrality ($\mathsf{RC}$), Jordan Center ($\mathsf{JC}$) and Dynamical Age ($\mathsf{DA}$),  when $G$ is a finite general tree with degree of each vertex less or equal to $d_m$, for $d_m=4,5,6,7$.}
\label{tab:compare_general}
\centering
\begin{tabular}{c|c|c|c|c|c}

 $d_m$& $\vert\kappa\vert$& Algo. \ref{algo:multi_algo} & $\mathsf{RC}$  & $\mathsf{JC}$ & $\mathsf{DA}$
\\
\hline
4 & 2.82 & 3.07  & 3.34 & 4.51 & \textbf{3.02}
\\
\hline
5 & 2.6  & \textbf{2.59} & 3.00  & 3.96  & \textbf{2.59}
\\
\hline
6 & 2.46 & \textbf{2.46}  & 2.85 & 3.87  & 2.51
\\
\hline
7 &2.45 & \textbf{2.33}  & 2.61  & 3.55  & 2.40

\end{tabular}
\end{table}

\color{black}

\section{Pseudo-Trees with a cycle}
\label{sec:cycle}
\begin{definition}
\label{def:pseudo_tree}
A pseudo-tree is a connected graph with equal number of vertices and edges, i.e., a tree plus an edge that creates a cycle.  
\end{definition}
In this section, we consider the special case where $G$ is a degree-regular graph, and $G_n$ has only a single cycle, i.e., $G_n$ is a pseudo-tree. We denote the cycle as $C_h$ where $h$ is the size (number of vertices on the cycle) of the cycle. Here, we call those vertices on $C_h$ \emph{cycle vertices}. Assume $v$ is a cycle vertex, then we define $t_{v}$ to be the subtree rooted at $v$ in $G_n$. Take Fig. \ref{fig:cycle_ex} for example, $t_{v_1}$ is the subtree that contains $v_1$, $v_4$ and $v_7$.  In this section, we study how a cycle affects the probability $P(v\vert G_n)$ when $G_n$ contains a cycle $C_h$. To generalize the analysis in \cite{who}, we should intuitively assume that the probability of being infected is proportional to the number of infected neighborhoods. With this assumption, the analysis in \cite{who} will not change, but we can consider the case that two infected vertices have a common susceptible neighborhood, i.e., there is a cycle in $G_n$.

\subsection{Impact of a Single Cycle On $P(G_n \vert v)$}

\begin{figure}
\begin{center}
\includegraphics[scale=0.3]{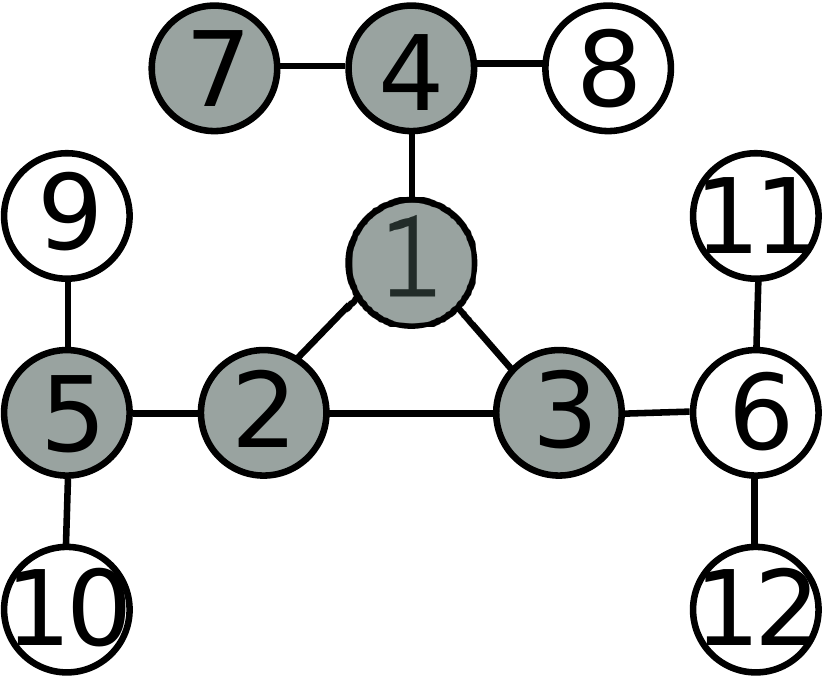}
\caption{$G_6$ is an infected subgraph with a single cycle $C_3$ containing three cycle vertices $v_1$, $v_2$ and $v_3$. We can partition $G_6$ into three subtrees say $t_{v_1}=\lbrace v_1,v_4,v_7\rbrace$, $t_{v_2}=\lbrace v_2,v_5\rbrace$, $t_{v_3}=\lbrace v_3\rbrace$.}
\label{fig:cycle_ex}    
\end{center}
\end{figure}

\begin{table}
\begin{center}
\caption{Numerical Example of $P(\sigma_i\vert G_6)$ using $G_6$ in Fig. \ref{fig:cycle_ex}}
\renewcommand\arraystretch{1.3}
\label{tab:cycle_ex}
\begin{tabular}{c|c|c}
$\sigma_i$ & Spreading Order & $P(\sigma_i\vert G_6)$
\\
\hline
$\sigma_1$ &  $v_4\rightarrow v_1 \rightarrow v_3 \rightarrow v_2 \rightarrow v_5 \rightarrow v_7$ & $\frac{2}{1200}$ 
\\
\hline
$\sigma_2$ & $v_4\rightarrow v_1 \rightarrow v_2 \rightarrow v_5 \rightarrow v_3 \rightarrow v_7$ & $\frac{2}{1800}$ 
\\
\hline
$\sigma_3$ & $v_4\rightarrow v_1 \rightarrow v_2 \rightarrow v_5 \rightarrow v_7 \rightarrow v_3$ & $\frac{2}{2520}$ 
\end{tabular}    
\end{center}
\end{table}

\begin{example}
\label{exam:cycle_af}
Consider the infected subgraph $G_6\subset G$ as shown in Fig. \ref{fig:cycle_ex}, where $G_6=\lbrace v_1, v_2, v_3, v_4, v_5, v_7\rbrace$ and there is a $3$-cycle in $G_6$. Consider a spreading order $\sigma_1\in M(v_4,G_6)$, where $\sigma_1: v_4\rightarrow v_1 \rightarrow v_2 \rightarrow v_3 \rightarrow v_5 \rightarrow v_7$. We have $P(\sigma_1\vert v_4)=(1/3)\cdot (1/4)\cdot (2/5)\cdot (1/4)\cdot (1/5)=2/1200$. Note that when $v_1$ and $v_2$ are infected, $v_3$ has two infected neighborhoods which implies that the probability of $v_3$ be infected in the next time period is twice higher than $v_5, v_7$ and $v_8$. In particular, there are three possible values for $P(\sigma_i\vert v_4)$ as shown in Table \ref{tab:cycle_ex}, for all $\sigma_i \in M(v_4,G_6)$. Moreover, we observe that the denominators are different in Table \ref{tab:cycle_ex}, due to sharing a common neighbor in the presence of a cycle. We call this property the \emph{cycle effect}.

\end{example}

Example \ref{exam:cycle_af} reveals some interesting properties of the cycle effect due to a single cycle:  
\begin{enumerate}
\item $P(\sigma_i \vert v)$ increases with how soon the last cycle vertex appears in $\sigma_i$ (as ordered from left to right of $\sigma_i$). For example, the last cycle vertex on $\sigma_1$ is $v_2$, and is $v_3$ on $\sigma_2$ and $\sigma_3$.
\item When there is a cycle in $G_n$, then $P(G_n\vert v)$ is no longer proportional to $\vert M(v,G_n)\vert$.
\item For each $\sigma_i$, there are actually two corresponding permitted spreading orders due to the cycle.
\end{enumerate}

The first property shows that $P(\sigma_i\vert v)$ is dependent on the position of the last cycle vertex in each spreading order. Note that the cycle effect is similar to the irregular effect, and the main difference lies in that all cycle vertices may cause the cycle effect instead of only one irregular vertex may cause the irregular effect. We proceed to compute $P(\sigma_i\vert v)$ as follows. For brevity of notation, let $v_l$ denote the last cycle vertex. 

\begin{definition}
\label{def:cycle_distance}
We let the distance from a vertex $v$ to the cycle $C_h$ denoted by $d(v,C_h)$ be defined by the minimum value of distances from $v$ to all cycle vertices on $C_h$. That is,  
\begin{center}
$d(v,C_h)=\min\limits_{v_i\in C_h} \{ d(v,v_i)\}$.
\end{center}
\end{definition}
 Take Fig. \ref{fig:cycle_ex} for example, let $C_3=\{v_1$, $v_2, v_3\}$, then  $d(v_7,C_3)=d(v_7,v_1)=2$ and $d(v_5,C_3)=1$.

Remark: For each $\sigma\in \vert M(v,G_n)\vert$, $v_l$ can be any vertex on the cycle $C_h$ except the vertex $v'$ with the distance $d(v,v')=d(v,C_h)$. Hence, there are $h-1$ possible $v_l$.

From previous observations, we have
\begin{equation}
\label{eq:cycle_mvgn}
\vert M(v,G_n)\vert=2\cdot\sum\limits^{n-t_{v_l}+1}_{k=d(v,C_h)+h} m^{v_l}_{v}(G_n,k)
\end{equation}
since the position of $v_l$ on the spreading order ranges from $d(v,C_h)+h$ to $n-t_{v_l}+1$. For example, in Table \ref{tab:cycle_ex}, we can see that $v_l=v_2$ is the $4$th element on $\sigma_1$ and $v_l=v_3$ is the $6$th element on $\sigma_3$ whence the order $4$ comes from $d(v_4,C_3)+h=1+3$, and the order 6 comes from $n-t_{v_l}+1=6-1+1$. Finally, the multiplication with $2$ is due to the third property.

Now, we can rewrite $P(G_n\vert v)$ for $G_n$ with a cycle as:
\begin{equation}
\label{eq:cycleP}
P(G_n\vert v)=\sum\limits^{n-t_{v_l}+1}_{k=d(v,C_h)+h} m^{v_l}_{v}(G_n,k)\cdot P^{v_l}_{v}(G_n,k),
\end{equation} 
and our goal is to find the vertex $\hat{v}$ that achieves 
\begin{equation}
P(G_n\vert \hat{v})=\max\limits_{v_i\in G_n}P(G_n\vert v_i).
\end{equation}
Since $P(G_n\vert v)$ is not proportional to $\vert M(v,G_n)\vert$, we should compute $P(G_n\vert v)$ by considering each part $m^{v_l}_{v}(G_n,k)$ and their corresponding probability $P^{v_l}_{v}(G_n,k)$. Let $z_d(i)=(i-1)(d-2)$, then
\begin{equation}
\label{eq:newcycleP}
P^{v_l}_{v}(G_n,k)=2\cdot\prod\limits^{k-1}_{i=1}\frac{1}{d+z_d(i)}\cdot\prod\limits^{n-2}_{i=k-1}\frac{1}{d+z_d(i)-1}.
\end{equation} 
The first factor in (\ref{eq:newcycleP}) is the probability that $k$ vertices are infected where the $k$th infected vertex is $v_l$, and the second factor is the probability of that all remaining $n-k$ vertices being infected thereafter. The $-1$ in the denominator of the second factor and the coefficient $2$ at the front are due to the common neighbor in a cycle. Note that multiplying by $2$ at the front makes no difference when computing $P(G_n\vert v)$ for each $v\in G_n$.  From (\ref{eq:cycleP}), we see that the number of spreading orders and the corresponding position of $v_l$ affect $P(G_n\vert v)$.

\subsection{Computing $\vert M(v,G_n)\vert$ for $v$ on $G_n$}
 In this section, we focus on computing $\vert M(v,G_n)\vert$. To compute $\vert M(v,G_n)\vert$, we can leverage the message-passing algorithm in \cite{who} if $G_n$ is a tree. Observe that for each infected vertex in $G_n$, it is infected by one of its infected neighbors (even if it has two infected neighbors), so the actual infecting route is a spanning tree of $G_n$ instead of a graph with cycle. Hence, the number of all spreading orders on a graph $G_n$ with a cycle can be computed as 

\begin{equation}
\label{eq:new_M}
\vert M(v,G_n)\vert=\displaystyle\sum_{1\leq i\leq h}\vert M(v,T_i)\vert ,
\end{equation} 
where $T_i$ is the spanning tree of $G_n$, for $i=1,2,...,h$. If $G_n$ contains a $C_h$, then the time complexity of computing $M(v,G_n)$ for $v\in G_n$ is $O(hn)$. Since $G_n$ has $h$ spanning trees, and for each spanning tree, we can use the message-passing algorithm in \cite{who} that has $O(n)$ time complexity.

\subsection{Epidemic Center on Uni-cyclic $G_n$}
In this section, we propose a theorem and a lemma to characterize the location of epidemic center in $G_n$. Instead of computing $\vert M(v,G_n)\vert$ for all $v\in G_n$ in each spanning tree, we leverage some analytical results to find $v_c$. Let $t_i=t_{v_i}$ be defined as above, and we slightly abuse the notation of the subtree size $\vert t_i \vert$ as $t_i$. 
\begin{theorem}
\label{thm:rough_locate_cycle}
Let $G$ be a degree regular graph, and $G_n$ be a subgraph of $G$ with a single cycle $C_h=\{v_1, v_2,...,v_h\}$. The epidemic center $v_c$ of $G_n$ satisfies one of the following condition:
\begin{enumerate}
    \item Each connected component of $G_n\backslash \lbrace v_c\rbrace$ is of size less or equal to $n/2$.
    \item $v_c$ is a cycle vertex and $t_i\leq n/2$ for $i=1,2,...,h$.
\end{enumerate}

\end{theorem}
Remark: Note that $v$ being epidemic center of $G_n$ does not mean that each connected component of $G_n\backslash \lbrace v\rbrace$ is of size less or equal to $n/2$. (Had $G_n$ been a tree, then this is true \cite{who}.)  However, condition 1 is the sufficient condition of a vertex being the epidemic center even in a general graph.

\begin{figure}
\begin{center}
\includegraphics[scale=0.35]{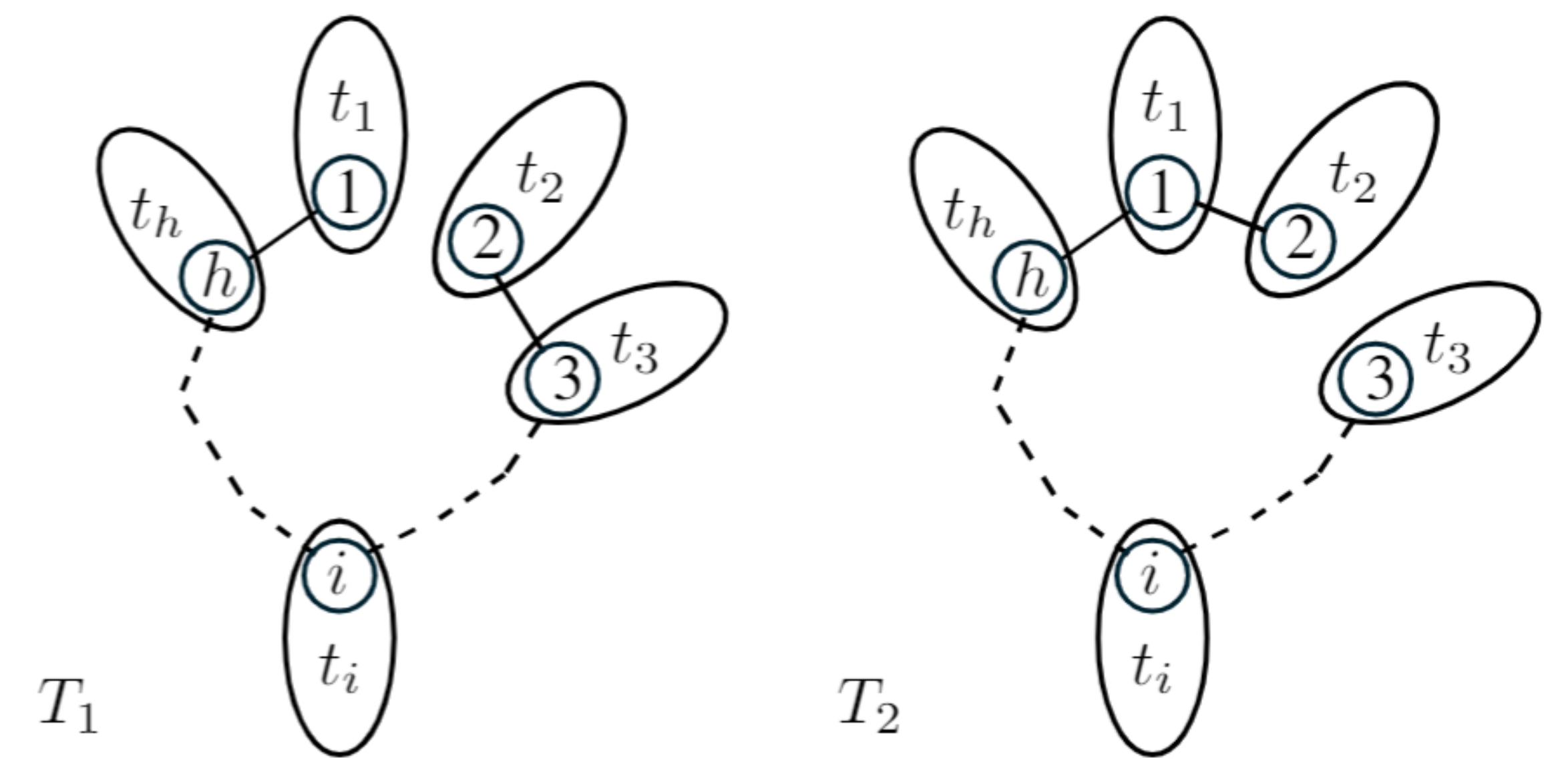}
  \caption{$C_h$ is constructed by $v_1,v_2,...,v_h$, and $t_i$ is a subtree rooted at $v_i$.}
\label{fig:spanning_tree_ratio}
\end{center}
\end{figure}

We can locate $v_c$ of $G_n$ by the first condition in Theorem \ref{thm:rough_locate_cycle}, however, if the first condition is not satisfied, then $v_c$ is on the cycle. In the following, we proceed to consider the case that if $v_c$ is a cycle vertex. Let $T_j$ denote the spanning tree of $G_n$ which is constructed by $G_n\backslash (v_j,v_{j+1})$, for $j=1,2,...,h-1$ and $T_h=G_n\backslash (v_h,v_1)$. Note that $(v_h,v_1)$ and $(v_j,v_j+1)$ for $j=1,2,...,h-1$ are cycle edges of $C_h$. 
\begin{proposition}
\label{thm:cycle_ratio}
Let $v_i$ be a cycle vertex and assume that $\vert M(v_i, T_p)\vert =r$, where $r$ and $p$ are integers and $1\leq p\leq h$. Then for $1\leq q\leq h$, we have
\begin{equation}
\label{eq:tree_ratio}
\frac{\vert M(v_i, T_q)\vert }{\vert M(v_i, T_p)\vert }= \frac{\prod\limits_{j\in C_h, j\neq i}T^i_{p,j}}{\prod\limits_{k\in C_h, k\neq i}T^i_{q,k}}, 
\end{equation} 
where $T^i_{p,j}$ is the subtree $T^i_{j}$ of the spanning tree $T_p$.
\end{proposition}

The ratio of $\vert M(v_i,T_p)\vert /\vert M(v_j,T_p)\vert$ is proportional to their branch size in $T_p$ if $v_i$ and $v_j$ are adjacent. Now, for the same vertex $v_i$, but in different spanning tree say $T_p$ and $T_x$, we can also derive the ratio $\vert M(v_i,T_p)\vert /\vert M(v_i,T_q)\vert$ from (\ref{eq:tree_ratio}).
 Hence, if we assume $\vert M(v_1, T_1)\vert =r$, then we can derive $\vert M(v, T)\vert$ for all $v\in C_h$ and $T$ is a spanning tree of $G_n$ in terms of $r$ and $t_i$, where $t_i$ is the subtree size as shown in Fig. \ref{fig:spanning_tree_ratio}. The time complexity to find $v_c$ is $O(n+h^2)$ in the worst case, where $h$ is the size of the cycle in $G_n$.

\begin{example}
Take $C_3$ for example, $\vert M(v_1,T_1)\vert /\vert M(v_1,T_2)\vert=t_3/(t_2+t_3)$. Assume $\vert M(v_1,T_1)\vert =r$, by (\ref{eq:tree_ratio}) we can conclude that
\begin{center}
$\vert M(v_1,G_n)\vert = \frac{2(t_2+t_3)}{t_3}r$; \\
$\vert M(v_2,G_n)\vert = \frac{2t_2(t_2+t_3)}{t_1t_3}r$; \\
$\vert M(v_3,G_n)\vert = \frac{2(t_2+t_3)}{t_1}r$, 
\end{center}
which implies 
\begin{center}
$\vert M(v_1,G_n)\vert : \vert M(v_2,G_n)\vert : \vert M(v_3,G_n)\vert = t_1 : t_2 : t_3.$
\end{center}
By Theorem \ref{thm:rough_locate_cycle} and Proposition \ref{thm:cycle_ratio}, we  conclude that if $G_n$ contains a $C_3$, then $v_c$ of $G_n$ is either a vertex that satisfies the first condition or a vertex $v_i$ on $C_3$ with $t_i=\max\limits_{1\leq j \leq 3} {t_j}$.
\end{example}

Lastly, we combine Lemma \ref{lem:key_lem} and Theorem \ref{thm:rough_locate_cycle} to characterize the location of the ML estimator on regular pseudo-tree. 

\begin{theorem}
\label{thm:MLE_cycle}
Let $G$ and $G_n$ be defined as in Theorem \ref{thm:rough_locate_cycle}. The optimal solution to (\ref{main_problem}) is either on the path from the epidemic center of $G_n$ to the cycle or on the cycle.
\end{theorem}
Remark: Theorem \ref{thm:MLE_cycle} is a combination of Theorem \ref{thm:rough_locate_cycle} and \ref{thm:keyfinite}. Besides, Theorem \ref{thm:MLE_cycle} generalizes the results in \cite{itw2017peter,asonam2017peter}.

\section{Algorithm for Finite Degree Regular Graph with Cycles}
This section proposes a novel distance-based algorithm to solve the epidemic source detection problem on a finite degree regular graph with cycles. From Theorem \ref{thm:keyfinite}, we can deduce that the likelihood of a vertex is greater if its distance to those irregular vertices and cycles is smaller. Hence, the maximum likelihood estimator should lie on the smallest induced subgraph containing three specific vertices: the epidemic center, the vertex closest to all cycles, and the vertex closest to all irregular vertices. Note that a degree on irregular vertex can be treated mathematically as a size one cycle. Hence we combine the irregular effect and the cycle effect in our algorithm. 
\begin{definition}
We say is a cycle is a \textit{minimum cycle} if there is no path between any two non-consecutive cycle vertices except the path along the cycle. 
\end{definition}
Since a vertex $v$ can be contained in multiple different-size cycles, we only take the minimum cycle that contains $v$ into consideration in our algorithm. Let $\mathcal{C}(v)$ denote the size of the minimum cycle containing $v$. If $v$ is not in any cycle and $deg(v)>1$, then we set $\mathcal{C}(v) =\infty$, otherwise  $\mathcal{C}(v)=1$. Note that when $deg(v)=1$, $v$ is regarded as a size $1$ cycle.

Base on Theorem 2,4, we heuristically define the weight $w_v$ of a vertex $v$ as 
\begin{equation}
w_v=\frac{\mathcal{C}(v)}{\mathcal{C}(v) +1}.
\end{equation}
Since we define the distance center to be the vertex with minimum distance centrality (cf. Definition \ref{def:distance}), we can design a weight such that the location of the ML estimator tends to be close to vertices with ``small weights''. This is also motivated by the fact that the likelihood of a vertex $v$ being the source is greater if $v$ has a larger epidemic centrality and is closer to those irregular vertices or cycles (cf. Theorem \ref{thm:keyfinite} and Theorem \ref{thm:MLE_cycle}). The definition of $\mathsf{SDC}(v,G_n)$ is a distance-based centrality. Furthermore, the definition of $w_v$ reveals that the irregular effect caused by a small cycle is greater than that caused by a large cycle which can be observed from Table \ref{tab:cycle_ex} and (\ref{eq:cycleP}). This definition implies that a vertex within a smaller cycle has a smaller weight which contributes ``more'' to $\mathsf{SDC}(v,G)$ while a vertex not in any cycle has weight $1$ which contributes ``less'' to $\mathsf{SDC}(v,G)$. Fig. \ref{fig:Weight_EX} illustrates such an example of a regular graph $G$ and the infected subgraph $G_{10}$ containing two different-sizes cycles, say $C_3$ and $C_4$. Note that, $v_{10}$ and $v_8$ have the same epidemic centrality, however, we have $P(G_{10}\vert v_{10})>P(G_{10}\vert v_8)$ since $v_{10}$ is closer to the smaller cycle than $v_8$.

\begin{figure}
\begin{center}
\includegraphics[scale=0.35]{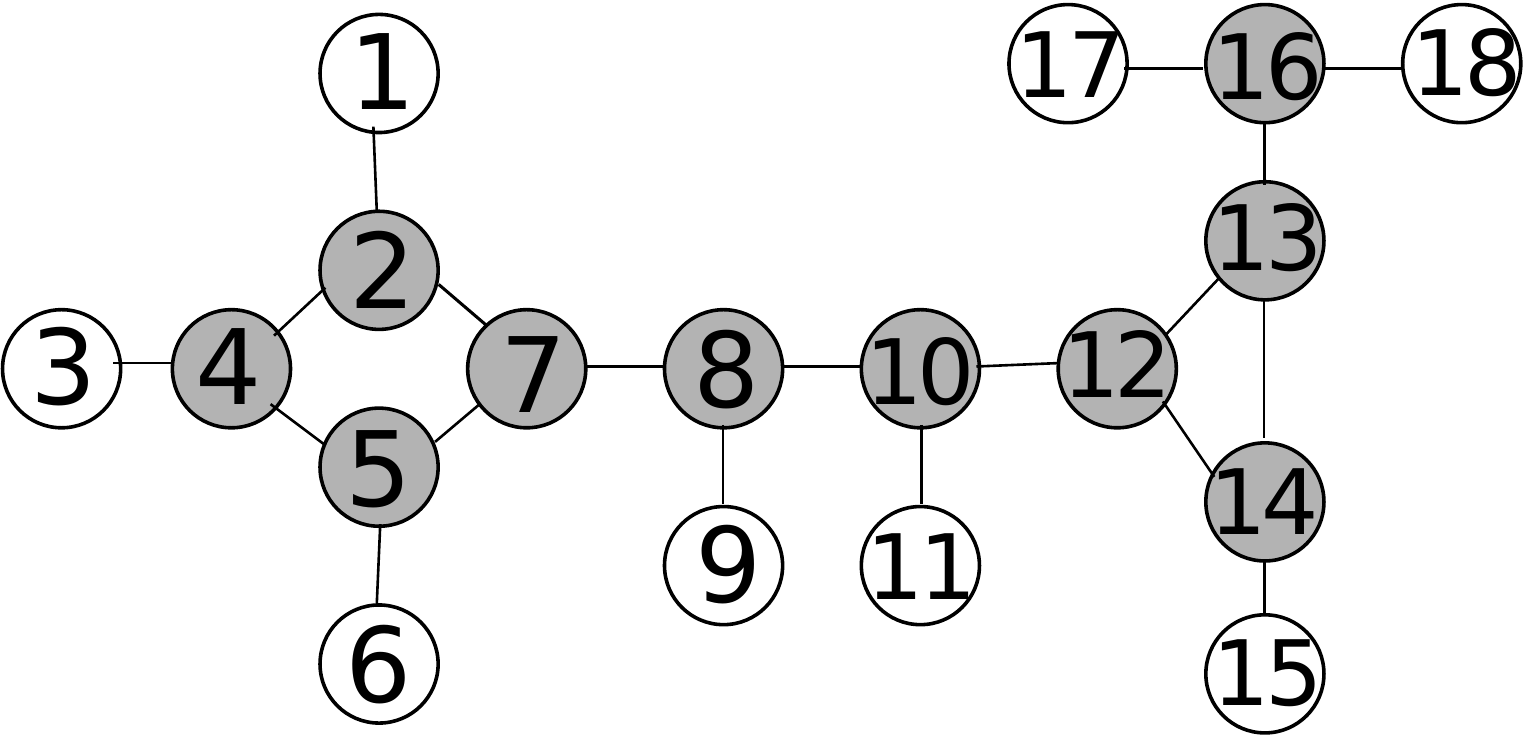}
  \caption{The infected subgraph $G_{10}$ contains ten nodes colored in grey. The epidemic centralities of $v_8$ and $v_{10}$ are the same, however, $v_{10}$ is the MLE of the true source in $G_{10}$. Since the irregular effect caused by a small cycle is greater than that of large cycle.}
\label{fig:Weight_EX}
\end{center}
\end{figure}

\begin{definition}
Given a $d$-regular graph $G$ and vertex $v$ of $G$, we define the \textit{statistical distance centrality} of $v$, $\mathsf{SDC}(v,G)$ as the summation of the weighted distance from $v$ to all other vertices in $G$. Hence, the statistical distance centrality of $v$ in $G$ is defined by 
\begin{equation}
\mathsf{SDC}(v,G_n)=\sum\limits_{ u \in G_n} w_u\cdot d(v,u).    
\end{equation}
The vertex $v_s$ with the minimum value for $\mathsf{SDC}(v,G)$ is called the \textit{statistical distance center}. 
\end{definition}

\begin{algorithm}[htb]
\caption{Statistical Distance-based Contact Tracing (SCT)}
\label{algo:master_algo}
\begin{algorithmic}
\Require 
$G_n$
\State Step 1: For each vertex $v$, compute the size $\mathcal{C}(v)$ of the minimum cycle containing $v$, and set $w_v=\frac{\mathcal{C}(v)}{\mathcal{C}(v) +1}$.
\State Step 2: For each vertex $v$, compute  $\mathsf{SDC}(v,G_n)$.
\State Step 3: Let $\hat{v}=\argmin\limits_{v\in G} \mathsf{SDC}(v,G)$.
\end{algorithmic}
\end{algorithm}

Algorithm \ref{algo:master_algo} is based on the idea of message passing. Let $lv(v)$ denote the level of $v$ in a $\mathsf{BFS}$ tree. In Step 2, for a given root $v_r$, we start a message-passing procedure in a $\mathsf{BFS}$ traversal to send a downward message containing level information to other vertices in the $\mathsf{BFS}$ tree. Upon receipt of this information, each leaf $v_l$ sends back an upward message containing $w_{v_l}\cdot lv(v_l)$ to its parent. Each internal vertex $v_{in}$ sends an upward message, containing the summation of all message from its children plus $w_{v_{in}}\cdot lv(v_{in})$, to its parent.   

In the following, we provide a time complexity analysis of Algorithm \ref{algo:master_algo}. For Step 1, the worst case time  complexity is $O(\vert \mathcal{C}_{min}\vert \cdot \vert E(G_n)\vert)$ \cite{Uno}, where $\vert\mathcal{C}_{min}\vert$ is the number of all minimum cycles in $G_n$. Since $G$ is a $d$-regular graph, each vertex in $G_n$ is contained in at most $d$ minimum cycles which implies $\mathcal{C}_{min} \leq d\cdot n$. The worst case time complexity for the Step 2 in the Algorithm \ref{algo:master_algo} is $O(n^3)$, since the $\mathsf{BFS}$ traversal for each vertex takes $O(n+\vert E(G_n)\vert)$. Hence, the worst case time complexity of Algorithm \ref{algo:master_algo} is $O(d\cdot n^3)$. 
In comparison with the $\mathsf{BFS}$ heuristic approach in \cite{who}, it applies the $\mathsf{BFS}$ traversal for each vertex and compute their epidemic centrality  which ends up with worst time complexity $O(n^3)$. 

\subsection{Experimental Performance Evaluation} 
We provide simulation results on different finite graphs with cycles, the first two simulations are conducted on synthetic graphs such as finite size grid graph and \textit{circulant graphs}. Both synthetic graphs are regular graphs with cycles except those vertices on the boundary of the grid graph. In synthetic graphs, we first simulate the virus spreading in a given network based on the model described in Section \ref{infinite}, to construct the infected subgraph $G_n$. Then, we apply Algorithm $\mathsf{SCT}$ to compute the source estimator. 

We conduct the other four experiments on real-world SARS-CoV2003 and COVID-19 contact tracing networks in Singapore and Taiwan. If we can identify the connection between any two confirmed cases in real-world contact tracing networks, we denote the connection (or contact) as an edge. However, when the number of confirmed cases is too large to record details of contact information, we can only have information about the geographical footprint for some confirmed cases. In this situation, we also denote those visited places as vertices, and we add an edge between a confirmed case and a place if the confirmed case had visited the place.

Since $G$ is unknown in practice and contact tracing networks are infected subgraphs $G_n$, we assume that $G$ is a regular graph with a few irregular vertices, and apply Algorithm $\mathsf{SCT}$ to the contact tracing networks to compute the source estimator. We use the graph distance from the actual source to the estimator to evaluate its performance.

\subsubsection{Grid Graph}
Disease spreading on grid graph is often considered under different spreading rules and models \cite{BB,balogh}. Hence, we select grid graph to be one of the testing synthetic networks. Simulation results are in Table \ref{tab:compare_grid} and one of the error distributions is in Fig. \ref{fig:compare_grid_10000-150}.  We can observe from Table \ref{tab:compare_grid} that the statistical distance based algorithm outperforms the $\mathsf{BFS}$ heuristic in \cite{who}. Moreover, the average error is increasing as the number of irregular vertices is increasing which again reveals the fact that the likelihood is evened out to those irregular vertices. 

\begin{figure}
\begin{center}
\includegraphics[scale=0.45]{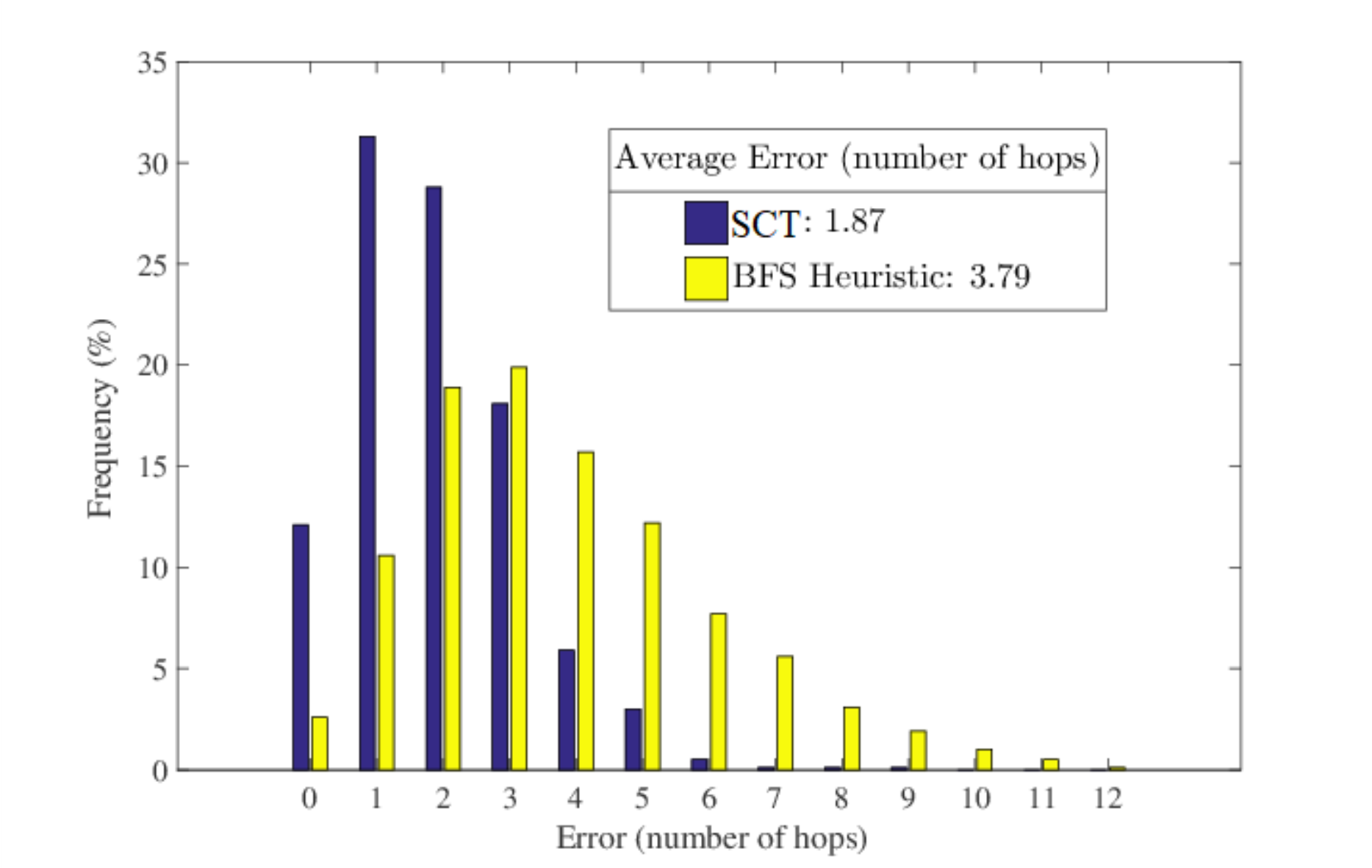}
\caption{Comparing the error distribution (in the number of hops) between Algorithm $\mathsf{SCT}$ and the $\mathsf{BFS}$ heuristic \cite{who} in a finite grid graph with $\vert G\vert =10000$ and $\vert G_n\vert =150$. In particular, the rate of the correct detection, i.e., $error=0$, is $12.1\%$ for Algorithm $\mathsf{SCT}$ and $2.6\%$ for the $\mathsf{BFS}$ heuristic.}
\label{fig:compare_grid_10000-150}
\end{center}
\end{figure}

\begin{table}
\caption{Average error (in terms of number of hops) comparing Algorithm $\mathsf{SCT}$ and $\mathsf{BFS}$ heuristic in \cite{who} when $G$ is a $100\times 100$ grid graph with different size of $G_n$.}
\label{tab:compare_grid}
\centering
\begin{tabular}{c|c|c|c|c}
 n &$\vert C_4\vert$& $\vert v_{ir} \vert$ & $\mathsf{SCT}$ & $\mathsf{BFS}$ heuristic \cite{who}
\\
\hline
150 &85.5 & 3.5 & 1.87 & 3.79 
\\
\hline
300& 199.5 & 7.5 & 2.33 & 6.11 
\\
\hline
500&364.2 &12.2 &3.14 & 8.37 
\\
\hline
800& 625.0 & 19.1 & 4.23 & 11.71 
\end{tabular}
\end{table} 

\subsubsection{Circulant Graph}
A circulant graph $G=(N; \mathsf{S})$ is a class of graphs which can be defined by its vertex set $V(G)=\mathbb{Z}$ and a set $\mathsf{S}$. The edge set is defined by $E(G)=\{ (v_i,v_j)\vert \text{ if }\vert i-j\vert \in \mathsf{S}\}$. Hence, if $\mathsf{S}$ is a subset of $[1, N/2)$ then $G$ is $(2\vert \mathsf{S}\vert)$–regular. Note that a circulant graph $G$ is connected if and only if $\mathsf{S}$ generates the integer group $Z_N$ and we only consider the connected circulant graph. In the simulation, we fix the size of $N$ and $\mathsf{S}$ and randomly choose integers from the interval $[1,N/2)$ to form the set $\mathsf{S}$. 

\begin{figure}
\begin{center}
\includegraphics[scale=0.45]{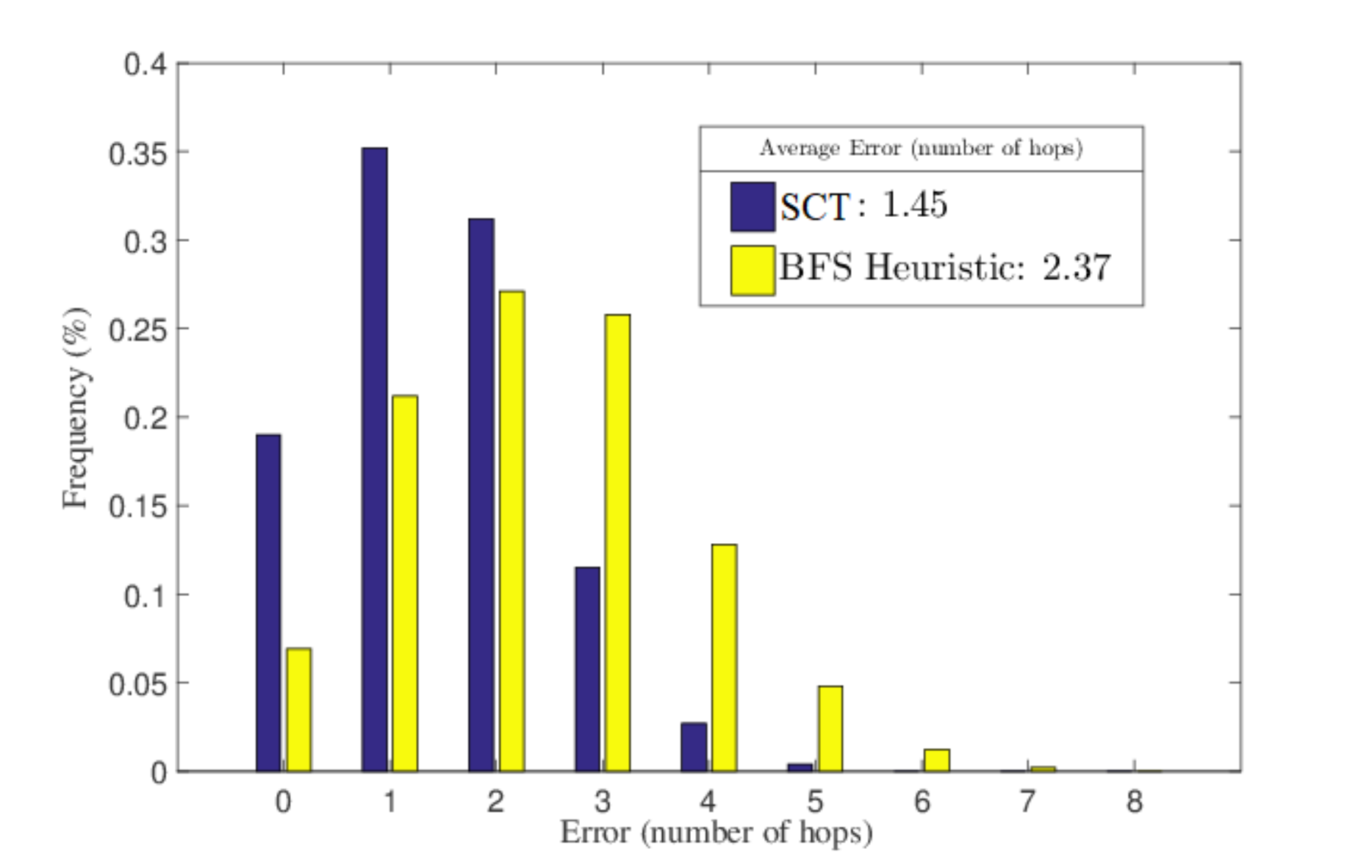}
\caption{Comparing the error distribution (in the number of hops) between Algorithm $\mathsf{SCT}$ and the $\mathsf{BFS}$ heuristic in \cite{who} on a circulant graph with $\vert G\vert =6000$ and $\vert G_n\vert =400$. In particular, the rate of the correct detection, i.e., $error=0$, is $19.1\%$ for Algorithm $\mathsf{SCT}$ and $7.0\%$ for the $\mathsf{BFS}$ heuristic.}
\label{fig:compare_circulant_6000-400}
\end{center}
\end{figure}

\begin{table}
\caption{Average error (in terms of number of hops) comparing Algorithm $\mathsf{SCT}$ and $\mathsf{BFS}$ heuristic in \cite{who} when $G$ is a random $d$-regular circulant graph with $|G|=6000$ and $\vert G_n \vert =400$.}
\label{tab:compare_circulant}
\centering
\begin{tabular}{c|c|c|c}
$d$   & $\#$(cycles)& $\mathsf{SCT}$ & $\mathsf{BFS}$ heuristic \cite{who}
\\
\hline
6 & 255.6 & 1.67 & 2.75
\\
\hline
 8 &  198.0& 1.45 & 2.37 
\\
\hline
10 & 167.1 & 1.33 & 2.07 
\\
\hline
12 & 147.7 & 1.24 & 1.97 
\end{tabular}
\end{table} 
\subsubsection{SARS-CoV2003 Contact Tracing Network in Taiwan}
We reconstruct the contact tracing network data of SARS-CoV2003 Taiwan from a graph, which indicates Potential bridges among hospitals and households, in \cite{ychen}. In the original data, there are four types of nodes which represent the confirmed case, suspected case, hospital, and area, respectively. Since cities or countries provide no information for personal contact, we delete all area nodes from the original data. In addition, we also delete all the nodes that represent suspected cases. We apply Algorithm $\mathsf{SCT}$ on this infected network and correctly identify the first place, Taipei Municipal Heping Hospital (now Taipei City Hospital Heping Branch), of cluster infection in April 2003 in Taiwan. In addition, the $\mathsf{BFS}$ heuristic approach chooses the red vertex, which represents a confirmed case (not the first case) who had been to Taipei Municipal Heping Hospital. The network graph is shown in Fig. \ref{fig:sars_tw}, and the orange vertex is the statistical distance center representing the Taipei Municipal Heping Hospital. 

\begin{figure}
    \centering
    \includegraphics[scale=0.4]{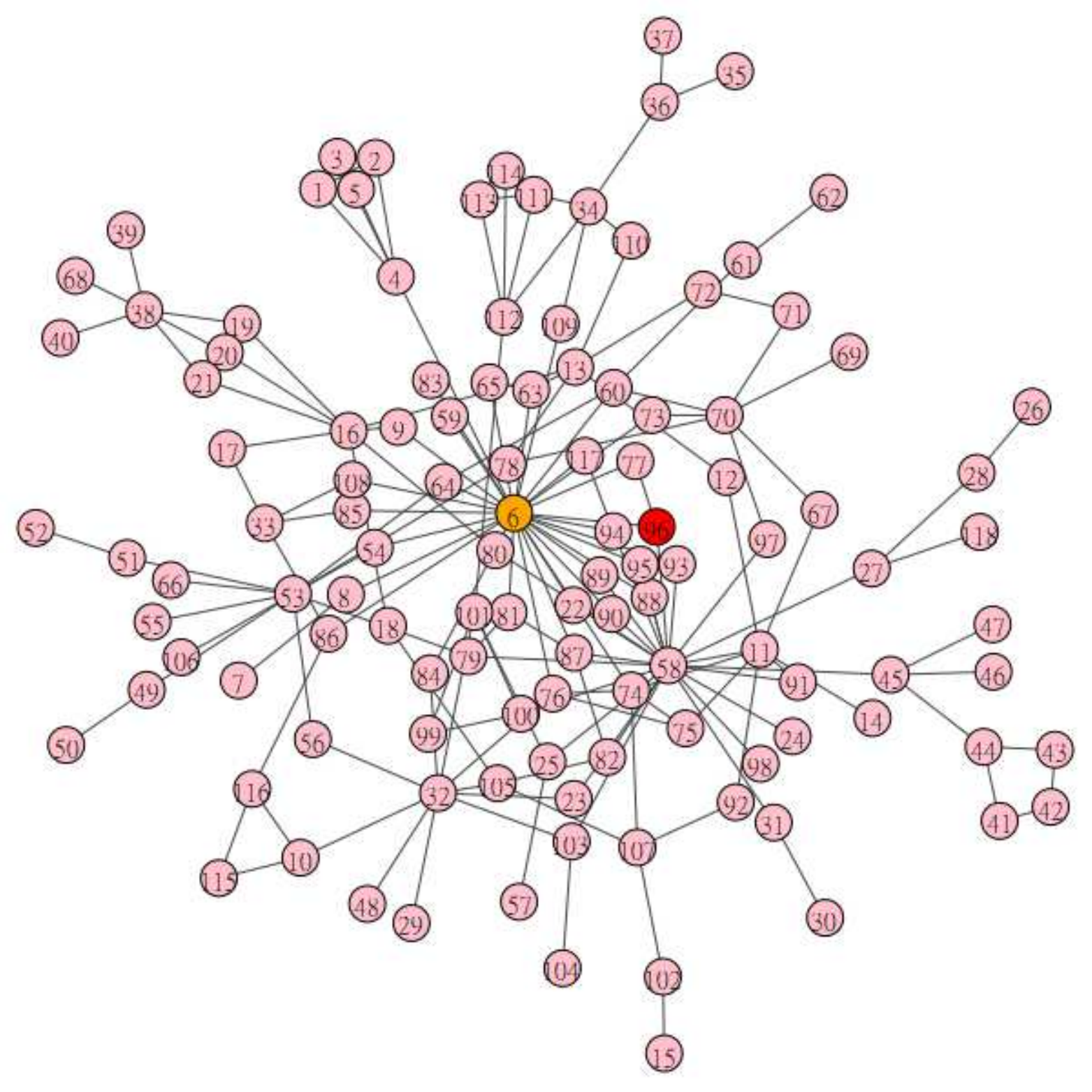}
    \caption{SARS-CoV2003 Contact Tracing Network in Taiwan. Each vertex represents either a confirmed case or a hospital. The orange and red vertices represent the source estimator determined by Algorithm $\mathsf{SCT}$ and  the $\mathsf{BFS}$ heuristic respectively. The orange vertex is the Taipei City Hospital Heping Branch where the major outbreak occurs and the red vertex is a confirmed case (not the first case) who had been to this hospital.}
    \label{fig:sars_tw}
\end{figure}

\subsubsection{COVID-19 Contact Tracing Network in Singapore, 2020 Mar-Apr}
The contact tracing network is an unconnected network due to the asymptomatic carriers, so we focus on the largest connected subgraph (cluster), including several worker dormitories and a construction site.
We apply Algorithm $\mathsf{SCT}$ on subgraphs of the contact tracing network in Singapore provided in \cite{sg_moh}. In our computation, each vertex represents either an infected person or a place where the person had visited. An edge between two vertices implies that either a person has visited a place or two places have at least one common visitor. Here we omit the edge of person-person contact since most of the contact history can only be traced back to a place, not a single person. Hence, we treat each person-vertex as a leaf vertex connecting to a place. 

The first massive outbreak occurred at the beginning of April and peaked on April 20. Hence, we consider the infected subgraph after April $1$. We define the source in the connected subgraph to be the first case in this connected subgraph. As far as we know, Case 655 attaching to Westlite Toh Guan (WTG) is the first case found in this subgraph on March 26. On April 3, the connected subgraph formed a 4-cycle, which is shown in Fig. \ref{fig:singapore}. We can apply Algorithm $\mathsf{SCT}$, which shows that WTG is the source estimator in this subgraph. After April 10, S11 Dorm (S11)  becomes the new epidemic center due to the link between STL and S11 Dorm (S11). Note that S11 contains the second earliest case in this cluster and becomes the largest cluster in Singapore, which has more than two thousand cases confirmed in the middle of May.

\begin{figure}
    \centering
    \includegraphics[scale=0.5]{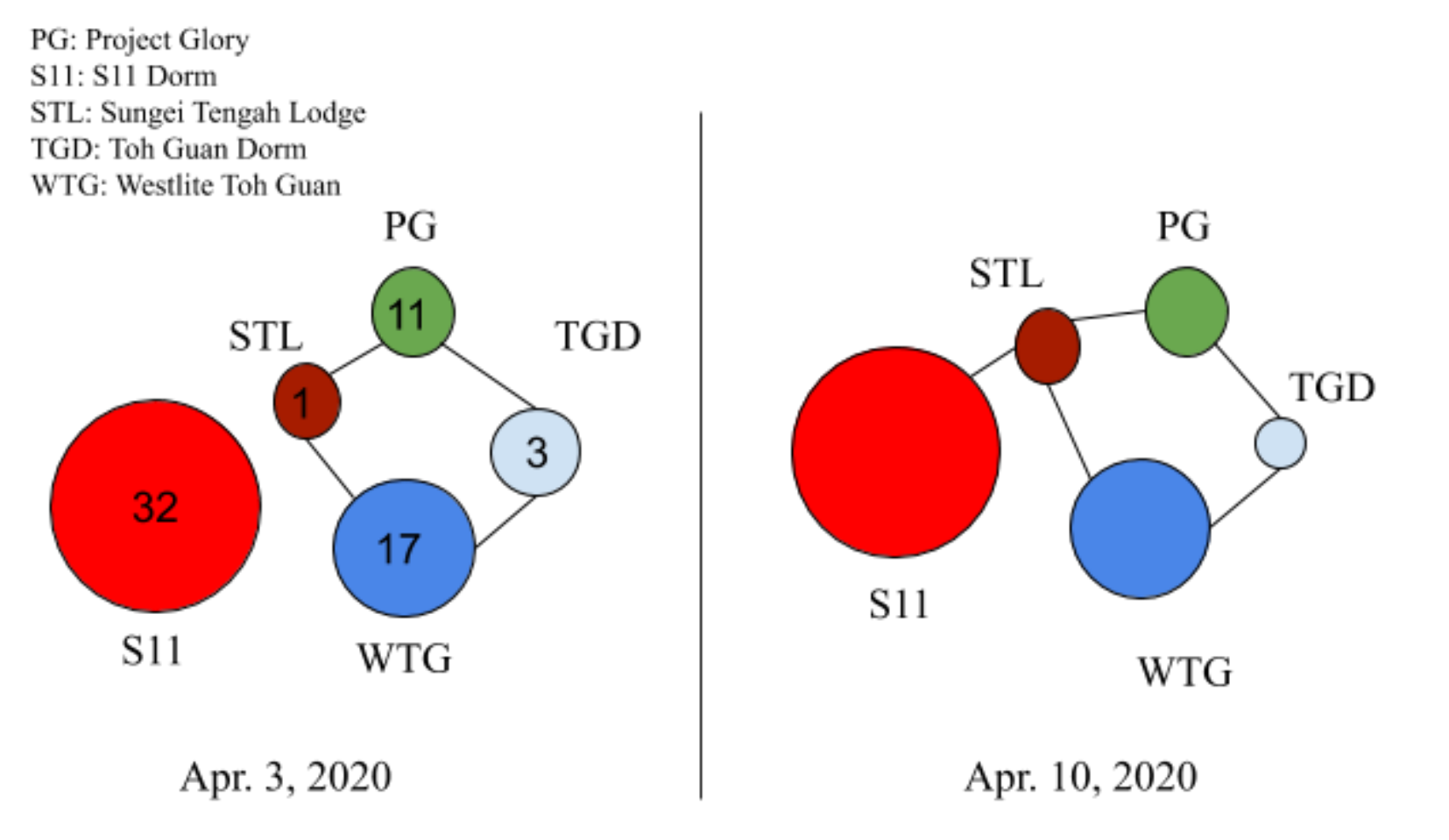}
    \caption{Each vertex is a cluster (place), and the number on each vertex is the total amount of infected people who have visited the place. On April 3, four places form a cycle, and all vertices are on the cycle. Algorithm $\mathsf{SCT}$ suggests that the source estimator is WTG where the first case in this subgraph comes from. On April 10, the epidemic center is S11 due to the link between S11 and STL.}
    \label{fig:singapore}
\end{figure}

\subsubsection{COVID-19 Contact Tracing Network in Taiwan, 2021 Feb and 2021 May}
We conduct experiments on two cluster infection in Taiwan recently. The first cluster infection originated from a northern Taiwan hospital on February. This network contains $18$ tractable domestic cases. The reason we select this networked data is that the contact network is public information provided by Central Epidemic Command Center in Taiwan \cite{tw_cdc} and the relation between cases in this network is clearly defined. Note that if we apply the $\mathsf{BFS}$ heuristic, then both case 838 and case 856 have the same possibility to be the source estimator. However, case 838 is the vertex with maximum statistical distance centrality, i.e., Algorithm $\mathsf{SCT}$ correctly identifies the first domestic case. 

The second cluster infection is the latest cluster infection found at the beginning of May 2021. As the source of this cluster is unknown, we let the source be the person with the earliest symptom onset. We collect the data before May 14,2021 from \cite{tw_cdc}, and apply the 2-mode network model \cite{ychen} to this cluster. Each vertex in this graph is either a workplace or a confirmed case. Both Algorithm $\mathsf{SCT}$ and the $\mathsf{BFS}$ heuristic identify the workplace, of the first case in this network. The contact tracing network is shown in Fig. \ref{fig:TW_covidMAY}, and the red vertex is the source estimator determined by both algorithms.

\begin{figure}
\begin{center}
\includegraphics[scale=0.3]{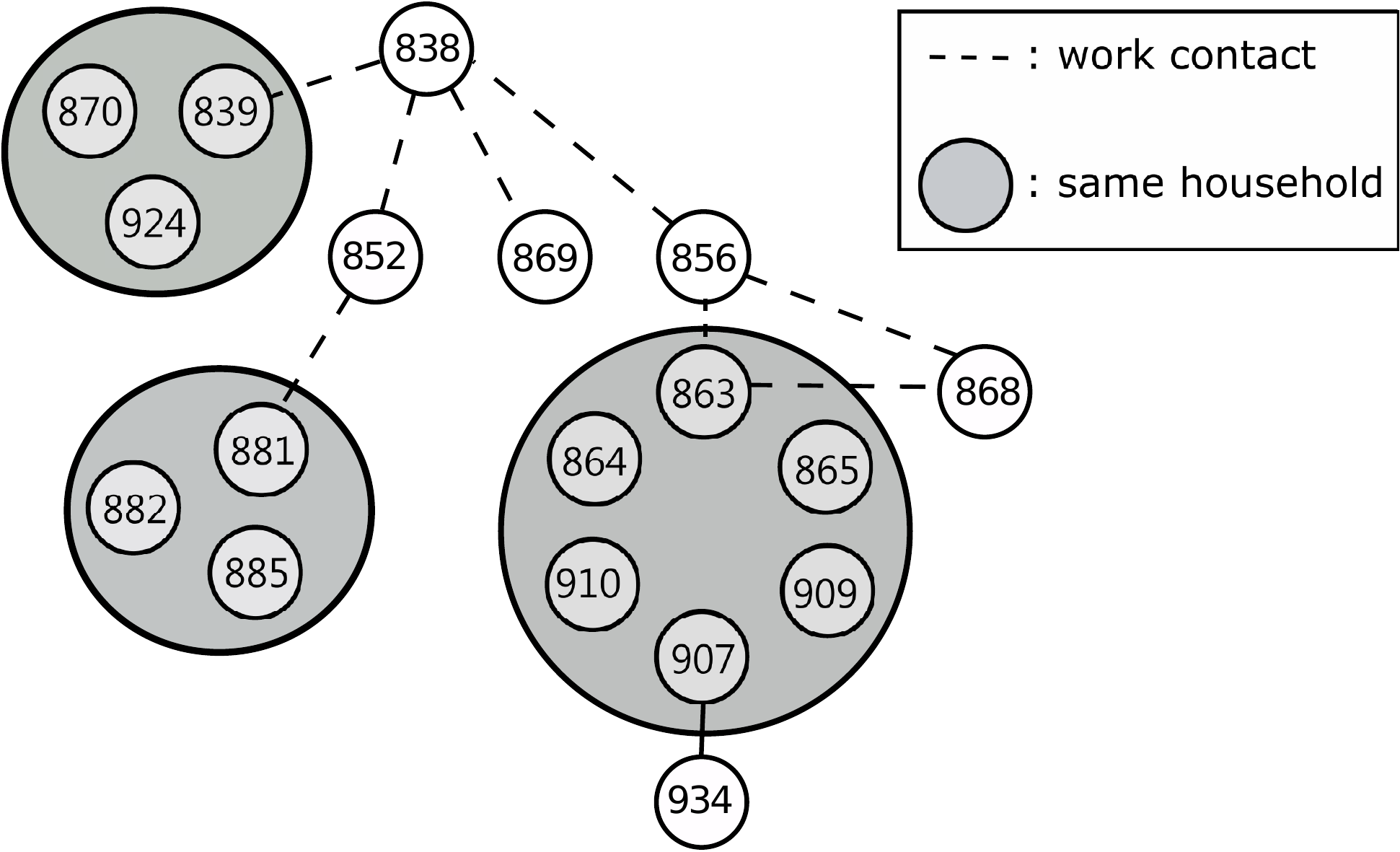}
\caption{The infected subgraph of the contact tracing network starting from a single case 838 in Taiwan. The number in each vertex is the case number. Case 838 has the maximum statistical centrality, and it is the first domestic case in this cluster.}
\label{fig:simulation_TW2021}
\end{center}
\end{figure}

\begin{figure}
\begin{center}
\includegraphics[scale=0.45]{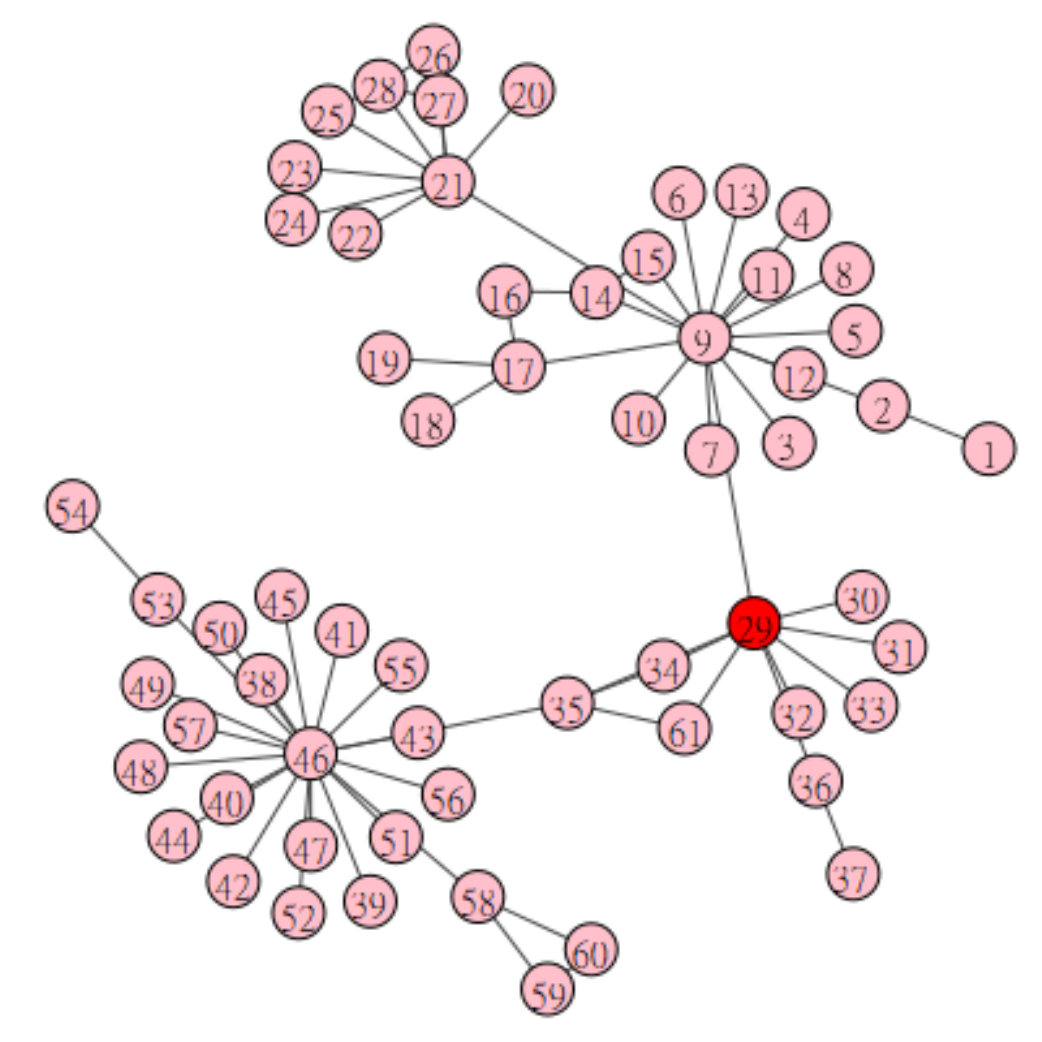}
\caption{The contact tracing network of COVID-19, at the beginning of May 2021, in Taiwan. The vertex in red is the work place of the person with the earliest symptom onset in this cluster.}
\label{fig:TW_covidMAY}
\end{center}
\end{figure}

\subsubsection{Simulation Results on Other Networks}
In addition to the rumor centrality approach, we have selected two other approaches, Dynamical Age \cite{spectral_rumor} and Jordan Centrality \cite{ZhuLei}, to compare to Algorithm $\mathsf{SCT}$. We use the average distance-error to measure the performance of each algorithm. The simulation results are shown in Table \ref{tab:all_simu_results}. In each simulation, we repeat the following process: generate $G_n$, find estimators and compute errors for five hundred times in each type of network. All datasets in Table \ref{tab:all_simu_results} are available at \cite{konect,snapnets} or can be generated by networkx \cite{SciPyProceedings_11}. Note that $\mathsf{SCT}$ works well in circulant graph and grid graph, since $\mathsf{SCT}$ is designed to solve the maximum likelihood estimation problem on finite regular graph with cycles. In addition to the selected algorithms, there are other approaches such as $\mathsf{PTVA}$ \cite{pinto} and gradient maximum likelihood algorithm ($\mathsf{GMLA}$) \cite{paluch} perform well on source detection problem. However, $\mathsf{PTVA}$ and $\mathsf{GMLA}$ are uncomparable with the methods listed in Table \ref{tab:all_simu_results}. Since $\mathsf{PTVA}$ and $\mathsf{GMLA}$ require additional information such as relative time stamps observed by observers. Moreover, to fairly compare all algorithms, the number of the observers and their location in the network also need to be carefully considered which is out of scope of this article.

\begin{table*}
\caption{Average error (in terms of number of hops) comparing Algorithm \ref{algo:master_algo} ($\mathsf{SCT}$) to Top-$k$ version of Rumor Centrality ($\mathsf{RC}$), Jordan Center ($\mathsf{JC}$) and Dynamical Age ($\mathsf{DA}$) in different networks}
\label{tab:all_simu_results}
\centering
\begin{tabular}{c|c|c|c|c|c|c}
 Network &$\vert G\vert$ &$\vert G_n \vert$& $\mathsf{SCT}$ & $\mathsf{BFS}$-$\mathsf{RC}$ & $\mathsf{JC}$ & $\mathsf{DA}$
\\
\hline
Circulant Graph(6000,6) & 6,000 & 400 & \textbf{1.67} & 2.75 & 2.87 & 1.90 
\\
\hline
$100\times 100$ Grid Graph &10,000 & 150 & \textbf{1.87} & 3.79 &  2.04 & 2.08
\\
\hline
Random $3$-regular Graph &5,000 & 200 & \textbf{1.26} & 1.42 & 1.57 & 1.32
\\
\hline
Barab\'asi-Albert(5000,3) & 5,000 & 300 & \textbf{2.74} & 4.25 & \textbf{2.75} & 2.96
\\
\hline
Canada Road Network & 1,965,206 & 100 & \textbf{3.15} & 3.41 & 3.19 & 4.03
\\
\hline
LastFM Asia Social Net. & 7,624 & 100 & \textbf{2.47} & 2.59 & 2.61 & 2.75
\\
\hline
Western U.S. Power Grid & 4,941 & 200 & \textbf{4.26} & 4.87 & 4.46 & 4.82
\end{tabular}
\end{table*} 
\color{black}
\section{Conclusion}
We present an optimal contact tracing algorithm for epidemic source detection that finds application in identification of superspreaders in an epidemic (e.g., the COVID-19 pandemic). Our algorithm design leverages a statistical distance centrality based on solving a graph-constrained maximum likelihood problem for contact tracing graphs with irregular vertices and cycles. As a high-dimensional statistical problem, epidemic source detection requires a careful understanding of the interaction between stochastic spreading processes and topological features like graph distances, cycles and finite boundaries, allowing us to resolve an open problem of finite degree-regular graphs with cycles in \cite{who}. Performance analysis demonstrated that our algorithm was near-optimal in correctly identifying superspreaders at some of the largest superspreading infection clusters using data from the SARS-CoV 2003 and COVID-19 pandemics in Singapore and Taiwan. It will be most interesting to study computational epidemiology through the lens of ``network centrality as statistical inference" to design robust contact tracing algorithms using mathematical and data-driven methods (e.g., machine learning techniques like deep learning, see \cite{bengio2020predicting,fei2021overview,siya2021CISS}) that can analyze past epidemic behaviors to fight against newly emerging epidemics.

\section{Acknowledgement}
We would like to acknowledge collaborations and interactions on this topic with H. Vincent Poor, Cheng-Shang Chang, Guanrong Chen, Siya Chen and Weng Kee Wong.




\begin{appendices}

\section{Proof of Theorem and Lemmas}
\subsection{Proof of Theorem \ref{thm:longterm}}
In this proof, we use the fact that 
\begin{center}
$P(v_c\vert G_n)>P_{v_c}^{v_e}(G_n,n)\cdot m_{v_c}^{v_e}(G_n,n)$,
\end{center}
and consider the ratio between the lower bound of $P(v_c\vert G_n)$ and $P(v_e\vert G_n)$ to simplify the proof. 
\begin{proof}
Let $G_n$ be the infected subgraph. Without loss of generality, we assume $n=2t+1$ and $d\geq 3$. For $v_e$, we have 
\begin{align*}
P(v_e\vert G_n)&=m_{v_e}^{v_e}(G_n,1)\cdot P_{v_e}^{v_e}(G_n,1) 
\\
&=1\cdot P_{v_e}^{v_e}(G_n,1)
\\
&=\prod\limits^{2t-1}_{i=0}\frac{1}{1+i(d-2)}=\prod\limits^{n-2}_{i=0}\frac{1}{1+i(d-2)}.
\end{align*}
For $v_c$, it is simpler to consider the last term of  (\ref{finiteP}) only, that is, $P_{v_c}^{v_e}(G_n,n)\cdot m_{v_c}^{v_e}(G_n,n)$. Note that $m_{v_c}^{v_e}(G_n,n)=\vert M(v_e,G'_{n-1})\vert$ where $G'_{n-1}=G_n\setminus\lbrace v_e\rbrace$. We have
\begin{align*}
&P_{v_c}^{v_e}(G_n,n)\cdot m_{v_c}^{v_e}(G_n,n)
\\
&= \left[ \prod\limits^{2t-1}_{i=0}\frac{1}{d+i(d-2)}\right] \cdot\frac{(2t)!}{2t(t-1)!t!} 
\\
&=\frac{(n-1)!}{(n-1)(\frac{n-3}{2})!(\frac{n-1}{2})!} \cdot\prod\limits^{n-2}_{i=0}\frac{1}{d+i(d-2)}.
\end{align*}  
Now, let us consider the ratio given by 
\begin{align*}
&\frac{P_{v_c}^{v_e}(G_n,n)\cdot m_{v_c}^{v_e}(G_n,n)}{P(v_e\vert G_n)}
\\
&=\frac{(n-1)!}{(n-1)(\frac{n-3}{2})!(\frac{n-1}{2})!}\cdot\frac{\prod\limits^{n-2}_{i=0}\frac{1}{d+i(d-2)}}{\prod\limits^{n-2}_{i=0}\frac{1}{1+i(d-2)}}
\\
&=c_1\cdot \frac{2}{(n-1)\cdot B(\frac{n-1}{2},\frac{n-1}{2})}\cdot\frac{\Gamma(n+\frac{1}{d-2}-1)}{\Gamma(n+\frac{d}{d-2}-1)}
\\
&\approx c_1\cdot \frac{2^{n-1}}{\sqrt{2\pi}}\cdot n^{-c_2},
\end{align*}
where $c_1$ and $c_2$ are some positive values with respect to $d$. The approximation is given by using Stirling's formula. The above result shows that the ratio becomes larger than 1 when $d$ is fixed, and $n$ is sufficiently large enough. This leads to
\begin{align*}
\frac{P(v_c\vert G_n)}{P(v_e\vert G_n)}>\frac{P_{v_c}^{v_e}(G_n,n)\cdot m_{v_c}^{v_e}(G_n,n)}{P(v_e\vert G_n)}>1,
\end{align*}
when $n$ is sufficiently large.
\end{proof}
\subsection{Proof of Lemma \ref{lem:key_lem}}

We first prove Lemma \ref{lem:d2_tree_ratio} and Lemma \ref{lem:key_lem}, which are tools to help us prove our main result.

\begin{lemma}
\label{lem:d2_tree_ratio}
Let $G$ be a tree and $v_a$, $v_b \in G$. If $d(v_a,v_b)=2$ and $M(v_a,G)>M(v_b, G)$, then $t^{v_b}_{v_a}>t^{v_a}_{v_b}$.
\end{lemma}
\begin{proof}
Let $v_m$ denote the vertex on the path from $v_a$ to $v_b$. Then we can express $M(v_a,G)$ as 
\begin{align*}
    M(v_a,G)&= \frac{t^{v_m}_{v_a}}{t^{v_a}_{v_m}}\cdot M(v_m,G),
\end{align*}
and $M(v_b,G)$ can be expressed in the same form.
Since $M(v_a,G)>M(v_b, G)$, we have
\begin{align*}
     \frac{t^{v_m}_{v_a}}{t^{v_a}_{v_m}} > \frac{t^{v_m}_{v_b}}{t^{v_b}_{v_m}}.
\end{align*}
Note that $t^{v_m}_{v_a}=t^{v_b}_{v_a}$ and $t^{v_m}_{v_b}=t^{v_a}_{v_b}$, we can rewrite the above inequality as 
\begin{align*}
     \frac{t^{v_b}_{v_a}}{t^{v_a}_{v_m}} > \frac{t^{v_a}_{v_b}}{t^{v_b}_{v_m}}.
\end{align*}
Moreover, we can leverage the fact  $n=t^{v_a}_{v_m}+t^{v_b}_{v_a}=t^{v_b}_{v_m}+t^{v_a}_{v_b}$ to replace $t^{v_a}_{v_m}$ and $t^{v_b}_{v_m}$ in the above inequality.
Hence, we have 
\begin{align*}
     \frac{t^{v_b}_{v_a}}{n-t^{v_b}_{v_a}} > \frac{t^{v_a}_{v_b}}{n-t^{v_a}_{v_b}},
\end{align*}
which implies $t^{v_b}_{v_a}>t^{v_a}_{v_b}$.

\end{proof}

\begin{proof}

To prove Lemma \ref{lem:key_lem}, we first consider the second condition, i.e., the case when $d(v_a, v_{ir})= d(v_b, v_{ir})=1$. Note that in this case, $v_a$ is not on the shortest path from $v_b$ to $v_{ir}$ and vice versa.

We consider the ratio of $m_{v_a}^{v_{ir}}(G_n, i)$ to $m_{v_b}^{v_{ir}}(G_n, i)$ for all possible $i$. Note that
\begin{align*}
  m_{v}^{v_{ir}}(G_n, i)=  {n-i \choose t^{v}_{v_{ir}}-(i-1)} M(v,t^{v_{ir}}_{v_a})\cdot M(v_{ir},t^{v}_{v_{ir}}),
\end{align*}
where $v=v_a,v_b$.
Since $M(v_a,t^{v_{ir}}_{v_a})$, $M(v_{ir},t^{v_a}_{v_{ir}})$, $M(v_b,t^{v_{ir}}_{v_b})$ and $M(v_{ir},t^{v_b}_{v_{ir}})$ are fixed when $G_n$ is given,  we have 
\begin{align*}
    \frac{m_{v_a}^{v_{ir}}(G_n, i)}{m_{v_b}^{v_{ir}}(G_n, i)}\propto \frac{{n-i \choose t^{v_a}_{v_{ir}}-(i-1)}}{{n-i \choose t^{v_b}_{v_{ir}}-(i-1)}}=\frac{{n-i \choose n-t^{v_{ir}}_{v_a}-1}}{{n-i \choose n-t^{v_{ir}}_{v_b}-1}}.
\end{align*}
Since $G_n$ is a tree and $d(v_a, v_{ir})= d(v_b, v_{ir})=1$, we have $d(v_a, v_b)=2$. By Lemma \ref{lem:d2_tree_ratio}, we have $t^{v_b}_{v_a}>t^{v_a}_{v_b}$, which implies $t^{v_{ir}}_{v_a}>t^{v_{ir}}_{v_b}$. Hence, the ratio $\frac{m_{v_a}^{v_{ir}}(G_n, i)}{m_{v_b}^{v_{ir}}(G_n, i)} $ is an increasing sequence with respect to $i$.  

Since $m_{v_a}^{v_{ir}}(G_n, 2)=m_{v_b}^{v_{ir}}(G_n, 2)$ and $\frac{m_{v_a}^{v_{ir}}(G_n, i)}{m_{v_b}^{v_{ir}}(G_n, i)} $ is  increasing with respect to $i$, we can conclude that $m_{v_a}^{v_{ir}}(G_n, i)\geq m_{v_b}^{v_{ir}}(G_n, i) $ which leads to
\begin{center}
    $\sum\limits_{i=d(v_a, v_{ir})+1}^{k}m_{v_a}^{v_{ir}}(G_n, i)\geq \sum\limits_{i=d(v_a, v_{ir})+1}^{k}m_{v_b}^{v_{ir}}(G_n, i)$.
\end{center}

On the other hand, we consider the case when $v_a$ is on the path from $v_b$ to $v_{ir}$, i.e., $\{v_a,v_b\}\in E(G)$ and $d(v_a, v_{ir})=d(v_b,v_{ir})+1$.

For brevity,  we denote the distance from $v_a$ to $v_{ir}$ as $D$, and relabel the vertices on the path from $v_p$ to $v_e$ as $u_0$ to $u_{D+1}$ which are shown in Fig. \ref{fig:relabel_path}.

\begin{figure}
\begin{center}
\label{fig:relabel_path}
\includegraphics[scale=0.4]{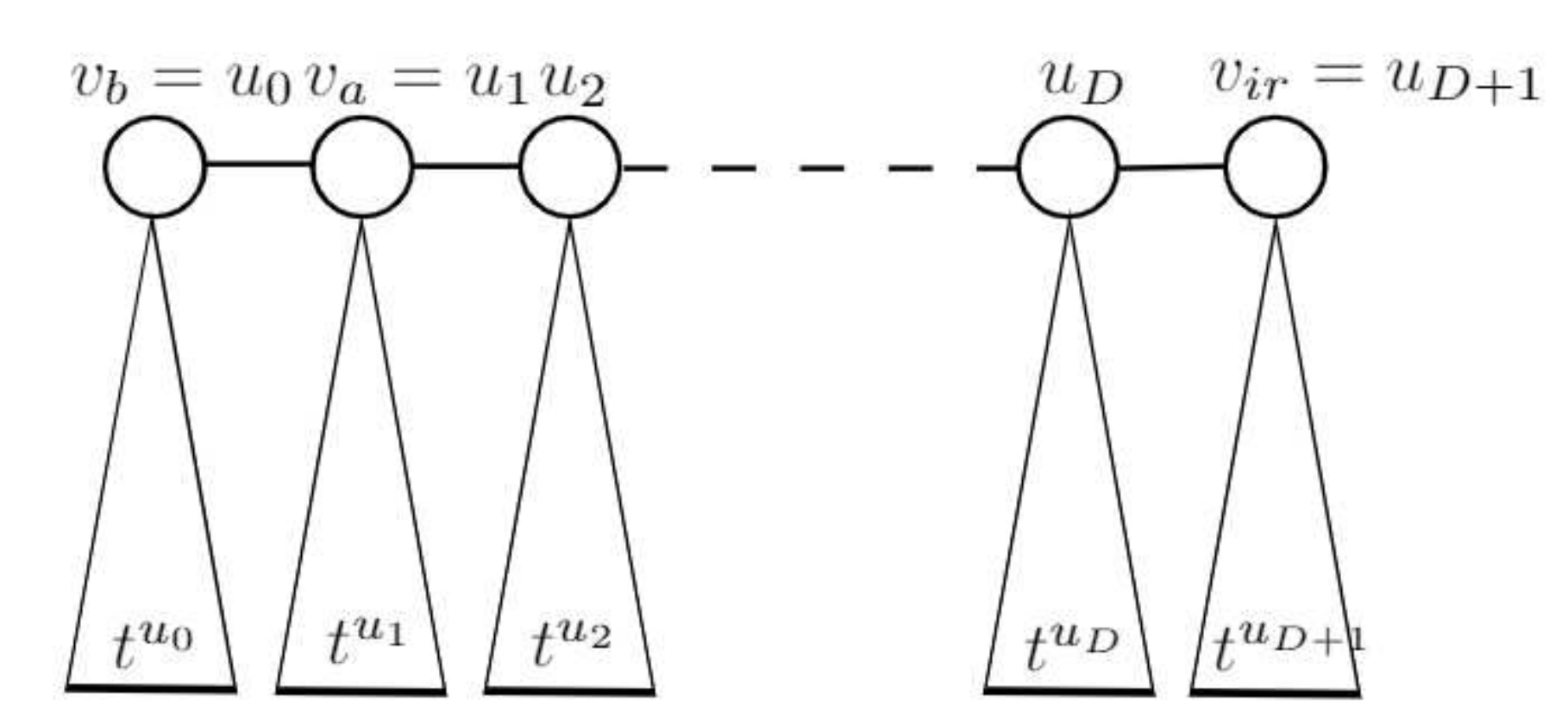}
\end{center}
\end{figure}

Again, we consider the ratio of $m_{u_1}^{u_{D+1}}(G_n, i)$ to $m_{u_0}^{u_{D+1}}(G_n, i)$ for all possible $i$. 
We can express  $m_{u_1}^{u_{D+1}}(G_n, i)$ and $m_{u_0}^{u_{D+1}}(G_n, i)$ as follows:
\begin{align*}
 & m_{u_1}^{u_{D+1}}(G_n, i)=  \\
  & \sum\limits_{1<i_1<\ldots < i_{D-1} <i} \Big[ \prod_{l=1}^{D-1} {n-i_l -t^{u_1}_{u_{l+1}} \choose t^{u_l}-1} \Big] {n-i \choose t^{u_1}_{u_{D+1}}-1}  \textbf{M}_a,
\end{align*}
where 
\begin{align*}
 \textbf{M}_a=M(u_1,t^{u_{D+1}}_{u_1})\cdot \prod_{l=2}^{D+1} M(u_l,t^{u_l}).
\end{align*}
We cam express $m_{u_0}^{u_{D+1}}(G_n,i)$ in the same way. Since $\textbf{M}_a$ is fixed when $G_n$ is given, we have

\begin{align*}
    \frac{m_{u_1}^{u_{D+1}}(G_n, i)}{m_{u_0}^{u_{D+1}}(G_n, i)} \propto \frac{ \sum\limits_{1<i_1<\ldots < i_{D-1} <i} \Big[ \prod_{l=1}^{D-1} {n-i_l -t^{u_1}_{u_{l+1}} \choose t^{u_l}-1} \Big]}{\sum\limits_{1<i_1<\ldots < i_{D} <i} \Big[ \prod_{l=1}^{D} {n-i_l -t^{u_1}_{u_{l+1}} \choose t^{u_l}-1} \Big]},
\end{align*}
which is a decreasing sequence with respect to $i$.

Note that when $i=D+1$, we have $m_{u_1}^{u_{D+1}}(G_n, D+1)>m_{u_0}^{u_{D+1}}(G_n, D+1)$.

Now, we can prove the main statement of Lemma \ref{lem:key_lem}. To contrary, suppose there is an integer $k$ such that 

\begin{center}
    $\sum\limits_{i=D+1}^{k}m_{v_a}^{v_{ir}}(G_n, i)< \sum\limits_{i=D+1}^{k}m_{v_b}^{v_{ir}}(G_n, i)$.
\end{center}
Since the ratio of $m_{u_1}^{u_{D+1}}(G_n, i)$ to $m_{u_0}^{u_{D+1}}(G_n, i)$ is a decreasing sequence, we have for all $k'>k$,  
\begin{center}
$\sum\limits_{i=D+1}^{k'}m_{v_a}^{v_{ir}}(G_n, i)< \sum\limits_{i=D+1}^{k'}m_{v_b}^{v_{ir}}(G_n, i)$.
\end{center}
This leads to $M(u_1,G_n)<M(u_0,G_n)$, which is a contradiction to the assumption. Hence, we can conclude that $\forall k$,

\begin{center}
    $\sum\limits_{i=D+1}^{k}m_{v_a}^{v_{ir}}(G_n, i)\geq \sum\limits_{i=D+1}^{k}m_{v_b}^{v_{ir}}(G_n, i)$.
\end{center} 
\end{proof}
\subsection{Proof of Theorem \ref{thm:keyfinite}}

\begin{proof}
Let $D=d(v_c,v_{ir})$ and we relabel vertices on the path from $v_c$ to $v_{ir}$ as $u_1, u_2, \ldots, u_{D+1}$ . We can partition $V(G_n)$ as follows
\begin{center}
  $V(G_n)=V(t^{v_{ir}}_{v_c})+V(t^{v_c}_{v_{ir}})+\sum\limits_{i=2}^{D} V(t^i)$.
  \end{center}
For any vertex $v$ in  $V(t^{v_{ir}}_{v_c})$ other than $v_{ir}$. If $d(v,v_{ir})>1$, then there is a neighbor of $v$, say $u$, such that $d(v,v_{ir})>d(u,v_{ir})$. Since $M(u,G_n)>M(v,G_n)$, by the condition 1 in Lemma \ref{lem:key_lem}, we can conclude that   
\begin{center}
$\sum\limits_{i=d(u, v_{ir})+1}^{k}m_{u}^{v_{ir}}(G_n, i)\geq \sum\limits_{i=d(u, v_{ir})+1}^{k}m_{v}^{v_{ir}}(G_n, i)$.
\end{center}
Since the probability $P(G_n\vert v)$ can be computed as follows
\begin{equation*}
P(G_n\vert v)=\sum\limits^{n-t^v_{v_{ir}}+1}_{i=D+1} m^{v_{ir}}_{v}(G_n,i)\cdot P^{v_{ir}}_{v}(G_n,i),
\end{equation*} 
and $P^{v_{ir}}_{v}(G_n,i)$ is a decreasing with respect to $i$.
We can conclude that $P(G_n\vert u)>P(G_n\vert v)$. This procedure can be continued until we reach to $v_{ir}$.
By the same argument, we can prove the case when $v$ is in other parts with the remaining uncomparable vertices being on the path from $v_c$ to $v_{ir}$.
\end{proof}

\subsection{Proof of Theorem \ref{thm:rough_locate_cycle}}
\begin{proof}
Let $v\in G_n$ and $v$ is not a cycle vertex and each connected component of $G_n\backslash \lbrace v\rbrace$ is of size less or equal to $n/2$. Consider any given spanning tree $T_j$ of $G_n$, assume that $v$ is not the rumor center of $T_j$. Then, there is a vertex $u$ such that $t^v_u>t^u_v$ on $T_j$. Since $t^v_u+t^u_v=n$, we can conclude that the size of the connected component of $G_n\backslash \lbrace v\rbrace$ containing $u$ is greater than $n/2$, which contradicts to the assumption. Hence, we have the fact that $v$ is the rumor center on each spanning tree $T_j$, for $j=1,2,\ldots , h$, and $v$ is also the rumor center of $G_n$ from (\ref{eq:new_M}).
\end{proof}

\end{appendices}


\vspace{-4em}
\begin{IEEEbiography}[{\includegraphics[width=1in,height=1.25in,clip,keepaspectratio]{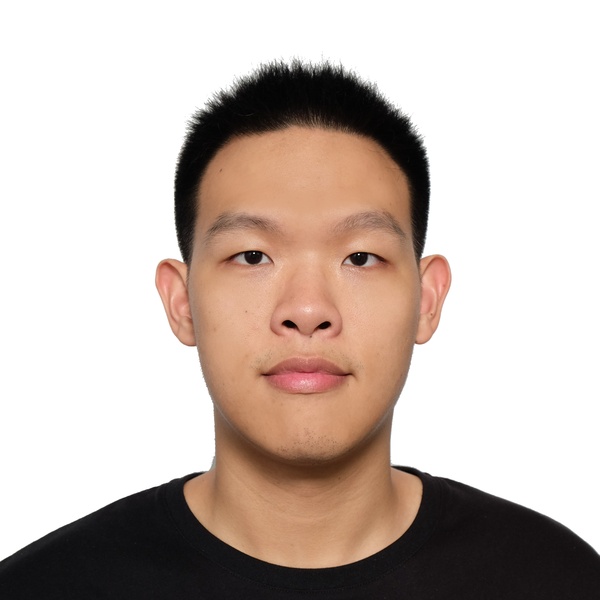}}]
{Pei-Duo Yu}  received the B.Sc. and M.Sc. degree in Applied Mathematics in 2011 and 2014, respectively, from the National Chiao Tung University, Taiwan. He received his  Ph.D. degree at the Department of Computer Science, City University of Hong Kong, Hong Kong. Currently, he is an Assistant Professor at the Department of Applied Mathematics in Chung Yuan Christian University, Taiwan. His research interests include combinatorics counting, graph algorithms, optimization theory and its applications.  
\end{IEEEbiography}

\vspace{-4em}
\begin{IEEEbiography}[{\includegraphics[width=1in,height=1.25in,clip,keepaspectratio]{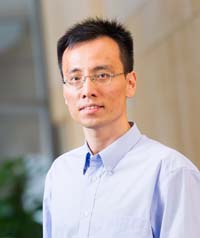}}]
{Chee Wei Tan}(M '08-SM '12)  received the M.A. and Ph.D. degrees in electrical and computer engineering from Princeton University, Princeton, NJ, USA, in
2006 and 2008, respectively. His research interests are in wireless networks, optimization and distributed machine learning. He is an IEEE Communications Society Distinguished Lecturer and has been an Associate Editor of IEEE Transactions on Communications and IEEE/ACM Transactions on Networking.
\end{IEEEbiography}

\vspace{-4em}
\begin{IEEEbiography}[{\includegraphics[width=1in,height=1.25in,clip,keepaspectratio]{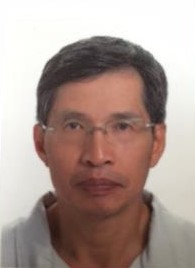}}]
{Hung-Lin Fu} (M '13) received the B. S. degree in Mathematics from National Taiwan Normal University in 1973 and Ph. D. degree in Mathematics (major in combinatorics) from Auburn University, Auburn, AL, in 1980. Before retirement, he was a Professor with the Department of Applied Mathematics, National Chiao Tung University, for 32 years. Currently, he is a fellow of the Institute of Combinatorics and Its Applications, and an Emeritus Professor of National Yang Ming Chiao Tung University, Hsinchu, Taiwan (since 2020). His interests are graph theory, combinatorial designs, and their applications.
\end{IEEEbiography}

\end{document}